\documentclass[11pt]{article}
\usepackage{a4,geometry}
\usepackage{graphicx}
\usepackage{amsmath,amssymb,amsthm,mathtools}
\usepackage{paralist}
\usepackage{bm}
\usepackage{xspace}
\usepackage{url}
\usepackage{fullpage, prettyref}
\usepackage{boxedminipage}
\usepackage{wrapfig}
\usepackage{ifthen}
\usepackage{mdframed}
\usepackage[table]{xcolor}
\usepackage{color}
\usepackage{transparent}
\usepackage{enumitem}
\usepackage{varwidth}
\usepackage{framed}
\usepackage{cancel}
\usepackage{nameref}
\usepackage{fullpage}
\usepackage{ifthen}
\usepackage{makecell}
\usepackage{multirow}

\usepackage{algorithm}
\usepackage[noend]{algpseudocode}
\algrenewcomment[1]{\(\triangleright\) {\small{\emph{\color{violet} #1}}}}
\algnewcommand{\LineComment}[1]{\State \(\triangleright\) \emph{\color{violet} \small #1}}
\algnewcommand{\Comment}[1]{\State \(\triangleright\) \emph{\color{violet} \footnotesize #1}}
\algnewcommand{\Invariant}[1]{\State \(\triangleright\) \emph{\color{red} #1}}
\usepackage[hypertexnames=false,pagebackref,letterpaper=true,colorlinks=true,pdfpagemode=UseNone,urlcolor=blue,linkcolor=blue,citecolor=violet,pdfstartview=FitH]{hyperref}
\usepackage[nameinlink]{cleveref}

\newtheorem{theorem}{Theorem}[section]

\newtheorem{lemma}[theorem]{Lemma}
\newtheorem{claim}[theorem]{Claim}
\newtheorem{corollary}[theorem]{Corollary}
\newtheorem{definition}[theorem]{Definition}

\newtheorem{fact}[theorem]{Fact}

\newtheorem{remark}{Remark}
\newtheorem{question}[theorem]{Question}

\newcommand{\ignore}[1]{}

\newcommand{\eq}{\leftarrow}

\newcommand{\cB}{\mathcal{B}}

\newcommand{\calB}{\mathcal{B}}

\newcommand{\NN}{\mathbb{N}}

\newcommand{\bone}{\mathbf{1}}

\DeclareMathOperator{\Pois}{Pois}
\newcommand{\h}{{\mathrm{H}\kern1pt}}
\newcommand{\I}{{\mathrm{I}\kern1pt}}

\newcommand{\poly}{\mathrm{poly}}
\newcommand{\polylog}{\mathrm{polylog}}

\newcommand{\ceil}[1]{{\left\lceil{#1}\right\rceil}}

\DeclareMathOperator*{\Exp}{\ensuremath{{\mathbf{E}}}}
\DeclareMathOperator*{\Var}{\ensuremath{{\mathbf{V}}}}
\DeclareMathOperator*{\Prob}{\ensuremath{\mathbf{P}}}
\renewcommand{\Pr}{\Prob}

\newcommand{\eps}{\ensuremath{\varepsilon}}

\newcommand{\Ot}{\ensuremath{\widetilde{O}}}

\colorlet{shadecolor}{blue!10}

\newenvironment{thm}{
	\begin{mdframed}[backgroundcolor=gray!20,topline=false,bottomline=false,leftline=false,rightline=false]\begin{theorem}
		}{
		\end{theorem}
	\end{mdframed}
}

\newcommand{\otilde}{\widetilde{O}}

\newcommand{\Sec}[1]{\hyperref[sec:#1]{\S\ref*{sec:#1}}} 
\newcommand{\Eqn}[1]{\hyperref[eq:#1]{(\ref*{eq:#1})}} 
\newcommand{\Fig}[1]{\hyperref[fig:#1]{Fig.\,\ref*{fig:#1}}} 
\newcommand{\Tab}[1]{\hyperref[tab:#1]{Tab.\,\ref*{tab:#1}}} 
\newcommand{\Thm}[1]{\hyperref[thm:#1]{Thm.\,\ref*{thm:#1}}} 
\newcommand{\Fact}[1]{\hyperref[fact:#1]{Fact\,\ref*{fact:#1}}} 
\newcommand{\Lem}[1]{\hyperref[lem:#1]{Lemma\,\ref*{lem:#1}}} 
\newcommand{\Prop}[1]{\hyperref[prop:#1]{Prop.~\ref*{prop:#1}}} 
\newcommand{\Cor}[1]{\hyperref[cor:#1]{Cor.~\ref*{cor:#1}}} 
\newcommand{\Conj}[1]{\hyperref[conj:#1]{Conjecture~\ref*{conj:#1}}} 
\newcommand{\Def}[1]{\hyperref[def:#1]{Definition~\ref*{def:#1}}} 
\newcommand{\Alg}[1]{\hyperref[alg:#1]{Alg.~\ref*{alg:#1}}} 
\newcommand{\Clm}[1]{\hyperref[clm:#1]{Clm.~\ref*{clm:#1}}} 
\newcommand{\Step}[1]{\hyperref[step:#1]{Step~\ref*{step:#1}}} 
\newcommand{\Exa}[1]{\hyperref[exa:#1]{Example~\ref*{exa:#1}}} 
\newcommand{\Obs}[1]{\hyperref[obs:#1]{Obs~\ref*{obs:#1}}} 

\newcommand{\dest}{{\tt deg-est}}

\usepackage{thmtools} 
\usepackage{thm-restate}

\newcommand{\putcolor}{\cellcolor{gray!25}}
 
\newcommand{\Tetek}{T\v{e}tek\xspace}

\newtoggle{anonymous}
\togglefalse{anonymous}

\title{Faster Estimation of the Average Degree of a Graph Using\\
	Random Edges and Structural Queries}
\date{}
\iftoggle{anonymous}{
\author{~}
}
{
\author{Lorenzo Beretta\thanks{Supported by a UC Santa Cruz start-up grant under Vaggos Chatziafratis.}\\
University of California, Santa Cruz\\
{\tt lorenzo2beretta@gmail.com}
\and
Deeparnab Chakrabarty\thanks{Supported by NSF-CAREER award 2041920.} \\
Dartmouth\\
{\tt deeparnab@dartmouth.edu}
\and
C. Seshadhri\thanks{Supported by NSF DMS-2023495, CCF-1740850, 2402572.} \\
University of California, Santa Cruz\\
{\tt sesh@ucsc.edu}
}
}
\begin{document}
\maketitle
\begin{abstract}

\noindent
We revisit the problem of designing sublinear algorithms for estimating the average degree of an $n$-vertex graph.
The standard access model for graphs allows for the following queries:
sampling a uniform random vertex, the degree of a vertex, sampling a uniform random neighbor of a vertex, and
``pair queries''  which determine if a pair of vertices form an edge. In this model,
original results [Goldreich-Ron, RSA 2008; Eden-Ron-Seshadhri, SIDMA 2019] on this problem
prove that the complexity of getting $(1+\varepsilon)$-multiplicative approximations
to the average degree, ignoring $\varepsilon$-dependencies, is $\Theta(\sqrt{n})$.
When random edges can be sampled, 
it is known that the average degree can estimated in $\widetilde{O}(n^{1/3})$ queries, even without pair queries [Motwani-Panigrahy-Xu, ICALP 2007; Beretta-\Tetek, TALG 2024]. \smallskip

We give a nearly optimal algorithm in the standard access model with random edge samples.
Our algorithm makes $\widetilde{O}(n^{1/4})$ queries exploiting the power of pair queries.
We also analyze the ``full neighborhood access" model wherein  
the entire adjacency list of a vertex can be obtained with a single query; this model is relevant in many practical applications.
In a weaker version of this model, we give an algorithm that makes $\widetilde{O}(n^{1/5})$ queries.
Both these results underscore the power of {\em structural queries}, such as pair queries and full neighborhood access
queries, for estimating the average degree.
We give nearly matching lower bounds, ignoring $\varepsilon$-dependencies, for all our results. \smallskip

So far, almost all algorithms for estimating average degree assume that the number of vertices, $n$, is known.
Inspired by [Beretta-\Tetek, TALG 2024], we study this problem when $n$ is unknown. 
With access to random edges, we give an algorithm that makes $O(n^{1/3})$ queries (which does not need pair queries), and prove that stronger structural queries do not help
in estimating average degree.
More generally, our results paint a nuanced picture of the complexity of estimating the average degree with respect to the access model.

%

\end{abstract}

\thispagestyle{empty}
\newpage
\tableofcontents
\thispagestyle{empty}
\pagebreak

\setcounter{page}{1}

\newcommand{\RV}{\mathbf{RandVert}}
\newcommand{\RE}{\mathbf{RandEdge}}
\newcommand{\RN}{\mathbf{RandNbr}}
\newcommand{\DEG}{\mathbf{Degree}}
\newcommand{\PAIR}{\mathbf{Pair}}
\newcommand{\CUT}{\mathbf{Cut}}
\newcommand{\add}{Additive}
\newcommand{\ADD}{\mathbf{Additive}}
\newcommand{\FN}{\mathbf{FullNbrHood}}
\newcommand{\id}{\mathsf{id}}
\newcommand{\davg}{d}

\newcommand{\EstSpar}{EstNumVert-or-EstDens}
\newcommand{\EstAvgDeg}{EstimateAverageDegree}
\newcommand{\EstNumEdges}{EstNumEdges}
\newcommand{\BTEstAvgDeg}{BT-EstAvgDeg}
\newcommand{\ERSEstAvgDeg}{ERS-EstAvgDeg}
\def\AlgEstDensHighDeg{{\sc EstDensHighDeg}}

\section{Introduction}

Estimating the average degree of a simple undirected graph 
is a fundamental problem in sublinear graph algorithms. Given an input
graph $G = (V,E)$ with $n=|V|$ and $m=|E|$, the average degree 
is $\davg := 2m/n$.
Feige~\cite{Feig06} gave a $(2+\eps)$-approximation to the average degree using $O(\eps^{-1}\sqrt{n})$ random (uar) vertices
and their degrees. He showed that if these are the only ways to access the graph, then 
no sublinear algorithm can get better than a $2$-approximation.
Goldreich and Ron~\cite{GR08} (henceforth called GR) showed that with access to random (uar) neighbors
of a vertex, one can obtain $(1 + \eps)$-approximations to the average degree\footnote{In dense graphs the number of queries is $\ll \sqrt{n}$, even if we did not know the graph was dense a priori.} 
    using $O(\poly(\eps^{-1}\log n)\sqrt{n/\davg})$ queries. Eden, Ron, and Seshadhri~\cite{ERS19} (henceforth called ERS)  gave a simple and elegant algorithm that uses $O(\eps^{-2}\sqrt{n/d})$ queries. Up to the dependence on $\eps$, this is known to be optimal~\cite{GR08}.

In the last two decades, both from theoretical and modeling considerations, other ``oracles/access models'' have been studied. Parnas and Ron~\cite{PaRo02} initiated the study of sublinear access models
for general graphs.
Kaufman, Krivelevich, and Ron~\cite{KaKrRo04} were the first to consider {\em adjacency queries}, often called {\em pair queries} (and we use this notation), that input a pair of vertices $x$ and $y$, returns $1$ if $(x,y)\in E$ and $0$ otherwise. Random vertex, random neighbor, 
degree, and pair queries form the {\em standard access} model for general graphs in sublinear algorithms. 
One can go beyond pair queries. Motivated by how large networks are typically crawled in practice via APIs
, Dasgupta, Kumar and Sarl\'os~\cite{DKS14} studied the model 
where a single query to a vertex $x$ returns the full adjacency list of $x$. Since then, this model has been used routinely in the applied community~\cite{CDKLS16,BEOF22,BKEBS25}.
 \Tetek and Thorup~\cite{TT22} performed a detailed study of estimating the average
 degree in the {\em full neighborhood access}.
While this power can improve the $\eps$-dependence, there is still an $\Omega(\sqrt{n/d})$ lower bound~\cite{TT22}.

Sampling uniform random edges is another natural way to access an unknown graph. To our knowledge, the paper~\cite{ABGPRY18} by Aliakbarpour, Biswas, Gouleakis, Peebles, Rubinfeld, and Yodpinyanee was the first to study sublinear graph algorithms
with this oracle. Random edge sampling has been a feature in the data mining community (e.g., see~\cite{LF06} and \cite{ANK13}), but understanding their power and limitations for sublinear algorithms is a more recent venture~\cite{EdRo18-2,AsKaKh19,TT22,BT24,BCM25}. 
Assadi, Kapralov, and Khanna~\cite{AsKaKh19} demonstrated the power of random edges for subgraph counting, and gave a detailed justification for this model.
Beretta and \Tetek~\cite{BT24} (henceforth BT) gave an elegant algorithm to estimate the average degree using $O(\eps^{-2}d)$ random edge samples~\cite{BT24}; together 
with ERS this implies a $O(\eps^{-2}\sqrt[3]{n})$ query algorithm with random edge, vertex, and neighbor samples to access the graph (along with degree queries).

Estimating the average degree with random vertex and edge sampling is a special case of estimating the average of a non-negative array 
with uar and proportional (to their value) samples.
Motwani, Panigrahy, and Xu~\cite{MPX07} gave a $O(\eps^{-3.5}\sqrt[3]{n})$ for obtaining
$(1+\eps)$-approximations~\cite{MPX07}
which was improved by BT to $O(\eps^{-4/3}\sqrt[3]{n})$~\cite{BT24}.
See \Cref{tab:estd-vertical} for a summaries of these bounds.
These papers also prove a $\Omega(\sqrt[3]{n})$ lower bound for average estimation of arbitrary non-negative arrays. 
This, however, does not rule out faster
 average degree estimation algorithms that utilize graph structure. This leads to the
main question driving our work:
\begin{center}
	\emph{Do structural graph queries (like pair queries / full neighborhood access model) and the ability to sample random edges \& vertices lead to faster algorithms for estimating average degree?}
\end{center}
\noindent
The answer is yes: we show that in the standard model, with the ability to sample random edges, we can get $(1+\eps)$-approximations to $d$ in $\Ot(\eps^{-2}n^{1/4})$ queries.
With full neighborhood access, we can estimate in $\Ot(\eps^{-2}n^{1/5})$ queries. Furthermore, all these query complexities are optimal up to $\poly(\eps^{-1}\log n)$ dependencies. 
These results are stated in \Thm{stan} and \Thm{advanced} respectively.
\medskip

\noindent
All (but one) of the above algorithms assume the knowledge of $n$, the number of vertices. BT~\cite{BT24} study the problem of estimating the average of a non-negative array
when the size $n$ is {\em not known} and
discover a qualitative difference:
with the knowledge of $n$, one can obtain $(1+\eps)$-approximations with $O(\sqrt[3]{n})$-samples,
but
without the knowledge of $n$, $\Omega(\sqrt{n})$ samples are necessary even for constant 
factor approximations. 
\begin{center}
	\emph{Do structural graph queries and the ability to sample random edges and/or vertices lead to faster algorithms for estimating average degree \underline{even when $n$ is unknown?}}
\end{center}
\noindent
Again, the answer is yes. But we also discover an interesting dichotomy. In the basic model (random vertex, degree, and
random neighbor queries) with random edges, we can get $O(\eps^{-2}\sqrt[3]{n})$ query algorithms. This matches the situation when $n$ was known, circumventing the BT $\Omega(\sqrt{n})$ lower bound.
We also show (\Thm{unknown-base}) that this query complexity is tight even if you allow full neighborhood access.
So, in the case of unknown $n$, when random edges samples are provided further structural queries don't help. This is starkly different from the known $n$ case.

If we cannot obtain random edge samples, then the picture is more nuanced. We show that with only random vertex and random neighbor queries, any good estimation algorithm
needs to make $\Omega(\sqrt{n})$ samples, even in {\em dense} graphs, if $n$ is unknown. Compare this with the $O(\sqrt{n/\davg})$-sample algorithms of GR and ERS.
If we allow full neighborhood access though, we can obtain $O(\eps^{-2}\sqrt{n/d})$ query algorithms matching the known $n$ results. Moreover, these algorithms can also estimate $n$,
faster than the trivial collision based estimator~\cite{RoTs16}.
All these query complexities are optimal, up to $\eps$ dependencies. \smallskip

\noindent
Our algorithmic results hold in slightly weaker models. We formally describe the models and our results
in~\Sec{results}.

\subsection{Model and formal statements of our results} \label{sec:results}

The input $G=(V,E)$ is a simple undirected graph, where $V$ is unknown. We set some terminology to describe the various models.

\begin{itemize}
	 \setlength\itemsep{0.01em}
	\item {\em Base Model.} This model allows access via the queries: (i) $\RV()$ that returns 
	a uniform at random (uar) vertex $x\in V$, (ii) $\RN(x)$ that takes input $x\in V$ and returns a uar neighbor $y\sim x$, 
	and (iii) $\DEG(x)$ that returns the degree of a vertex $x$. All of these take unit time/cost.
	The algorithms of GR~\cite{GR08} and ERS~\cite{ERS19} work in this model.
	
	\item {\em Standard Model.} Along with the base model queries, 
	the standard model also allows (iv) $\PAIR(x,y)$ which returns $1$ if $(x,y)\in E$ or $0$ otherwise, in unit time/cost.
	This is the usual
	model for most sublinear graph algorithms.
	
	\item {\em Advanced Model.} The advanced model allows for all queries in the base model. 
    Furthermore, instead of using the full neighborhood access model, we introduce
	an intermediary model with {\em additive} queries: given a {\em subset} $S\subseteq V$,
	$\ADD(S)$ returns the number of edges in $S$, but in $|S|$ time/cost. 
	Note that it is a natural generalization of pair queries (where $|S| = 2$), and can clearly be simulated by $|S|$ full neighborhood access
	queries. Thus, this model is (potentially strictly) weaker than the full neighborhood access model.
	
	\item {\em Random Edges Samples.} The $\RE()$ oracle returns an edge $e\in E$ uniformly at random in unit time/cost.
\end{itemize}
\noindent
{\bf On additive queries.} We borrow the term ``additive'' from the literature on Combinatorial Search, in particular the paper 
of Grebinski and Kucherov that uses such queries to {\em reconstruct} a graph~\cite{GK00}. A big difference is that these works often assume unit cost for an additive query
irrespective of the size of $S$; our cost grows linearly in the size. Note that $O(|S|^2)$ many $\PAIR$ queries can simulate an $\ADD(S)$ query. 
Additive queries are intrinsically related to ``cut queries''.
If we know $\deg(x)$ for all $x\in S$, then cut queries can simulate $\ADD(S)$ and vice-versa. 
We also note the connection with Independent Set (IS) queries studied in the literature~\cite{BeHa+20,ChLeWa20,AdMcMu22}, and elaborate
more in~\Sec{related}. \smallskip


\noindent
\begin{remark}
Throughout the paper, we assume $G$ has no isolated vertices, that is, the minimum degree is $\geq 1$. This is a standard assumption made
in sublinear graph algorithms. One can weaken this to assume a constant fraction of the vertices have degree $\geq 1$. 
\end{remark}
\noindent
{\bf Results when $n$ is known.} 
When $n$ is known, we show the power of $\RE$ queries.
Recall that BT~\cite{BT24} has an algorithm to estimate $(1\pm \eps)$-approximation to $\davg$ in $O(\eps^{-2}\davg)$ queries. This algorithm plays a balancing role in all our final complexities.
%
%
All algorithms stated below are randomized and all statements made about them hold with probability $> 0.99$.

\begin{thm}\label{thm:stan}
	Consider the standard model with $n$ known and $\RE$ queries.
	
	\noindent
	(a) There is an algorithm that outputs a $(1 + \eps)$-approximation to $\davg$ in $\Ot\left(\eps^{-2}\min\left(\davg, \sqrt[3]{n/\davg}\right)\right)$ queries.
	In particular, we get an $(1 + \eps)$-approximation to $\davg$ in $\Ot\left(\eps^{-2}n^{1/4}\right)$ queries.
	
	\noindent
	(b) Furthermore, any algorithm in this setting that outputs a constant factor approximation to $d$ 
	requires $\Omega(\min(d, \sqrt[3]{n/d}))$ queries, which is $\Omega(n^{1/4})$ for worst-case $d$.
\end{thm}

\noindent
Note that our query complexity is smaller in denser graphs, just as the algorithms of GR~\cite{GR08} and ERS~\cite{ERS19}.
We also prove (\Cref{thm:lb-known-n}, part (a)) that the base model with random edge samples (so no structure queries)
cannot improve upon the $\sqrt[3]{n}$ algorithms of~\cite{MPX07,BT24}, and so the above theorem shows that structure helps significantly.
The next theorem shows that when we have $\ADD()$ queries,
the exponent becomes even smaller. 

\begin{thm}\label{thm:advanced}
	Consider the advanced model with $n$ known and $\RE$ queries.

\noindent
(a) There is an algorithm that outputs a $(1+ \eps)$-approximation to $\davg$ in $\Ot\left(\eps^{-2}\min\left(\davg, \sqrt[4]{n/d}\right)\right)$ queries.
In particular, we get an $(1+ \eps)$-approximation to $\davg$ in $\Ot\left(\eps^{-2}n^{1/5}\right)$ queries.

\noindent
(b) Furthermore, any algorithm in this setting (even with the stronger full neighborhood access) 
that outputs a constant factor approximation to $d$ 
requires 
$\Omega(\min(d, \sqrt[4]{n/d}))$ queries, which is $\Omega(n^{1/5})$ for worst-case $d$
\end{thm}
\noindent
In summary, if we consider the base, standard, and advanced model respectively, then the complexity
of estimating $d$ is $n^{1/3}, n^{1/4}, n^{1/5}$ respectively, and
all these bounds are optimal in $n$. (We summarize all the results in \Cref{tab:estd-vertical} with references
to the specific theorems proving each bound.)\smallskip

\noindent
{\bf Results for unknown $n$.}
When $n$ is not known, we cannot apply the algorithms of GR~\cite{GR08} and ERS~\cite{ERS19}; indeed, we prove 
a lower bound (\Cref{clm:apple}) of $\sqrt{n}$ samples for any constant approximation algorithm in the base model with unknown $n$, even when the graph is dense. 
Note that one can always get a good estimate 
of $n$ in $O(\sqrt{n})$ samples~\cite{RoTs16} and run GR or ERS. We give improved algorithms 
with access to either $\RE$ queries, or $\ADD$ queries.

\begin{thm}\label{thm:unknown-base}
		Consider the base model with $n$ unknown and $\RE$ queries. 
		
		\noindent
		(a) There is an algorithm that outputs a $(1+ \eps)$-approximation to $\davg$ in $\Ot\left(\eps^{-2}\min\left(\davg, \sqrt{n/\davg}\right)\right)$ queries.
		In particular, we get an $(1+ \eps)$-approximation to $\davg$ in $\Ot\left(\eps^{-2}n^{1/3}\right)$ queries.
		
		\noindent
		(b) Furthermore, any algorithm, even in the full adjacency model with random edge samples, that outputs a constant factor approximation to $d$ 
		requires
        $\Omega(\min(d, \sqrt{n/d}))$
        queries, which is $\Omega(n^{1/3})$ for worst-case $d$.
\end{thm}
\noindent
To reiterate, random edge samples alone recover GR and ERS bounds that become smaller if the graph is dense. 
Additional structural queries don't help. This is in stark contrast with the known $n$ setting, where we get a hierarchy of complexities.

If we do not have random edge samples but have strong structural queries, we can still recover GR and ERS style results in the case of unknown $n$. 
%
%
%
%
%

\begin{thm}\label{thm:unknown-no-edge}
		Consider the advanced model with $n$ unknown.
	
	\noindent
	
	\noindent
	(a) There is an algorithm that outputs a $(1 + \eps)$-approximation to $\davg$ in $\Ot\left(\eps^{-2}\sqrt{n/\davg}\right)$ queries.

	\noindent
	(b) Furthermore, any algorithm, even in the stronger full adjacency model, that outputs a constant factor approximation to $d$ 
	requires $\Omega(\sqrt{n/\davg})$ samples.
\end{thm}
\noindent
One may ask about algorithms in the standard model (base + pair queries) when $n$ is unknown. 
We show (\Cref{cor:unknown-n-no-edge})  that one can get $O\left(\eps^{-2}\min(\sqrt{n}, n/\davg)\right)$- query algorithms, and this is tight (\Cref{clm:cherry}).
This is a somewhat weak improvement over the base model, since it is only better when $d \gg \sqrt{n}$.
\Cref{tab:estd-vertical} summarizes our findings. Up to $\poly(\eps^{-1}\log n)$ dependencies, the picture
is complete and gives a nuanced understanding of the complexity of estimating $\davg$. (We do consider
resolving the $\poly(\eps^{-1}\log n)$ factors important, and discuss more in \Sec{related} and \Sec{future}.)

\begin{table}[h]
\centering
\begin{tabular}{|c|c|c|}
\hline
\multicolumn{3}{|c|}{$n$ is known} \\
\hline
Model & with $\RE$ & w/o $\RE$ \\
\hline\hline
\multirow{2}{*}{\bf Base} 
& {\footnotesize $\otilde(d \land \sqrt{n/d})= \otilde(n^{1/3})$ \cite{MPX07, ERS19, BT24}} 
& {\footnotesize $O(\sqrt{n/d})$ \cite{ERS19}} \\ 
& \putcolor {\footnotesize $\Omega(d \land \sqrt{n/d})= \Omega(n^{1/3})$ (\Thm{lb-known-n}(a))} 
& \\ \hline
\multirow{2}{*}{\bf Standard} 
& \putcolor {\footnotesize $\otilde(d \land \sqrt[3]{n/d})= \otilde(n^{1/4})$ (\Thm{stan})} 
& \\ 
& \putcolor {\footnotesize $\Omega(d \land \sqrt[3]{n/d})= \Omega(n^{1/4})$ (\Thm{lb-known-n}(b))} 
& \\ \hline
\multirow{2}{*}{\bf Advanced} 
& \putcolor {\footnotesize $\otilde(d \land \sqrt[4]{n/d})= \otilde(n^{1/5})$ (\Thm{advanced})} 
& \\ 
& \putcolor {\footnotesize $\Omega(d \land \sqrt[4]{n/d})= \Omega(n^{1/5})$ (\Thm{lb-known-n}(c))} 
& {\footnotesize $\Omega(\sqrt{n/d})$~\cite{TT22}} \\ \hline
\end{tabular}

\vspace{1em}

\begin{tabular}{|c|c|c|}
\hline
\multicolumn{3}{|c|}{$n$ is not known} \\
\hline
Model & with $\RE$ & w/o $\RE$ \\
\hline\hline
\multirow{2}{*}{\bf Base} 
& \putcolor {\footnotesize $O(d \land \sqrt{n/d}) = O(n^{1/3})$ (\Thm{unknown-base})} 
& {\footnotesize $O(\sqrt{n})$ \cite{RoTs16}} \\ 
&  & \putcolor {\footnotesize $\Omega(\sqrt{n})$ (\Clm{apple})} \\ \hline
\multirow{2}{*}{\bf Standard} 
&  & \putcolor {\footnotesize $O(\min(\sqrt{n}, \frac{n}{d}))$ (\Cor{unknown-n-no-edge})} \\ 
&  & \putcolor {\scriptsize $\Omega(\min(\sqrt{n},\frac{n}{d}))$ (\Clm{banana})} \\ \hline
\multirow{2}{*}{\bf Advanced} 
&  & \putcolor {\footnotesize $O(\sqrt{n/d})$ (\Thm{unknown-no-edge})} \\ 
& \putcolor {\footnotesize $\Omega(d \land \sqrt{n/d}) = \Omega(n^{1/3})$ (\Clm{date})} 
& \putcolor {\footnotesize $\Omega(\sqrt{n/d})$ (\Clm{banana})} \\ \hline
\end{tabular}

\caption{\small \em Complexity of estimating $d$ in various models considered in this paper. 
The top table corresponds to the case when $n$ is known, and the bottom one when $n$ is not known.
All cells in gray are new results of this paper. All cells without text are implied by other cells.
For clarity, we hide the dependence on $\eps$. All our algorithms have a dependence of $1/\eps^2$.}
\label{tab:estd-vertical}
\end{table}

%

\noindent
{\bf On estimating $n$.}
When $n$ is unknown, it is natural to ask whether $n$ itself can be approximated. There is always
a simple collision-based estimator that only uses $O(\sqrt{n})$ uar vertex samples (see, for example, Theorem 2.1 in this paper~\cite{RoTs16} by Ron and Tsur). This procedure
works for any discrete universe and has nothing to do with graphs. But consider the
example of a cycle on $n$ vertices. Addition structural queries or uar edge samples
cannot help in beating an $\Omega(\sqrt{n})$ lower bound. Still, one can ask whether this trivial ``birthday paradox" bound can be beaten
when the graph density is higher. Indeed, we can replicate the results for estimating $\davg$. In particular, we have the following.

\begin{thm} \label{thm:est-n}
	Consider the advanced model with $n$ unknown.
	
	\noindent
	
	\noindent
	(a) There is an algorithm that outputs a $(1\pm \eps)$-approximation to $n$ in $\Ot\left(\eps^{-2}\sqrt{n/\davg}\right)$ queries.

	\noindent
	(b) Furthermore, any algorithm, even in the stronger full adjacency model, that outputs a constant factor approximation to $n$ 
	requires $\Omega(\sqrt{n/\davg})$ samples.
\end{thm}
%
%
\noindent
As in the case of estimating $\davg$, we also show that in the base model, any algorithm requires $\Omega(\sqrt{n})$ queries, and in the standard model, we pin down the complexity to be $\Theta(\min(\sqrt{n}, n/d))$ queries.
On the other hand, random edge samples alone don't give any benefit for estimating $n$ in any of the
models; we prove a $\Omega(\sqrt{n})$ lower bound (\Cref{clm:cherry}) for constant approximation of $n$ in the base model with random edge samples.\smallskip

\noindent
{\bf Without $\RV$ queries.}
Finally, we consider the setting where we access the graph only via random edges, neighbors, but {\em not} random vertices. This was a question\footnote{It is not clear if the question in~\cite{BT24} allowed random neighbor access; however, if we allow full neighborhood access (which is stronger than $\ADD$) we get the power of querying random neighbors as well.} asked in~\cite{BT24}.
We show that $\ADD(S)$ queries lead to some power, and without them we have the same $\Omega(\sqrt{n})$ lower bound~\cite{MPX07,BT24} as for estimating the sum of an arbitrary vector.

\begin{thm}\label{thm:onlyedges}
	Consider access to the graph only via random edge samples, random neighbor samples but no random vertex samples, and $\ADD$ queries. 
	
	\noindent
	(a) There is an algorithm that outputs a $(1\pm \eps)$-approximation to $\davg$ in $\Ot\left(\eps^{-2}\min\left(\davg, \sqrt{n/\davg}\right)\right)$ queries.
	In particular, we get an $(1\pm \eps)$-approximation to $\davg$ in $\Ot\left(\eps^{-2}n^{1/3}\right)$ queries.

		\noindent
	(b) Furthermore, any algorithm in this setting that outputs a constant factor approximation to $d$ (even with the stronger full neighborhood access)
	requires $\Omega(n^{1/3})$ samples. Moreover, if instead of $\ADD$ we had only $\PAIR$ queries, then any such algorithm requires $\Omega(\min(\sqrt{n},n/\davg))$ queries.
\end{thm}

\noindent
The algorithm is described in~\Cref{sec:3.3} while the lower bound is proved in~\Cref{clm:elderberry}.

\subsection{Main ideas} \label{sec:ideas}

In this section we discuss the main technical ideas behind our theorems. 
For the sake of exposition, we drop all dependencies on $\eps$ and when we refer to ``estimating" a parameter,
we mean a multiplicative $(1+\eps)$-approximation. \medskip

\noindent
{\bf When $n$ is known.} We begin with a discussion of the ideas behind \Thm{stan} and \Thm{advanced}.
Recall, that the $O(d)$ query algorithm is the one from BT~\cite{BT24}, and we now discuss how to get
$O(\sqrt[3]{n/d})$ and $O(\sqrt[4]{n/d})$ query algorithms in the standard and the advanced models, respectively. The ideas behind both algorithms is the same.
It involves putting together a bunch of algorithmic pieces where each piece is relatively
simple, but they need to carefully ``mesh" together to get the (near) optimal sample complexity.

As is standard in many sublinear graph algorithms, the pain point is a subset of ``heavy'' (high-degree) vertices
since they lead to high variance of many estimators. To make this a bit more precise, 
let $p$ be a tune-able parameter which is set later. Let $H$
be the top $p$-fraction of vertices by degree, and $\Delta$ be the minimum
degree in $H$. With $O(1/p)$ uar vertex samples and degree queries,
one can approximately estimate $\Delta$. \medskip

\noindent
{\em ``Easy" Case.} Suppose $H$ induces at most a constant fraction of the edges of the whole graph. 
In this case, one can show that $\Delta$ forms a coarse approximation to the average degree; in particular,  $d = O(\Delta)$. 
In this situation, we can apply the idea of the ``harmonic estimator''~\cite{BT24} from BT which samples a vertex $x$ proportional to its degree (which is possible to do using random edge samples)
and then returns the reciprocal of the degree if $x\in H$, and otherwise returns $0$.  It follows directly from the results~\cite{BT24} that $O(1/p)$ samples suffice to estimate $\davg$ (\Cref{lem:harmest}). \smallskip

\noindent
{\em The more ``challenging" case.}  The challenging case is when $H$, the set of heavy vertices, induces most of the edges.
In this case the harmonic estimator has too much variance. To address this case we do the following.
\begin{asparaitem}
	\item We identify subsets $A\subseteq H$ and $B\subseteq H$ that are quasi-regular, that is, 
	every vertex in $A$ (and respectively $B$) has their degrees within a multiplicative factor $2$. Moreover, they induce a $\polylog n$ fraction of all the edges
	between them. That is, $m(A,B) \geq m/\polylog n$. This follows from standard binning techniques. (\Cref{lem:recogij})
	\item Since $m(A,B)$ is so ``large'', with $\polylog n$ edge samples, we can estimate $\frac{m(A,B)}{m}$ (\Cref{lem:ber-subs}). So, if we can estimate $m(A,B)$, 
	then we can estimate $m$. 	
	\item Using a rejection sampling technique, we can sample random vertices from $A$ (or $B$) using the power to sample random edges. (\Cref{lem:rand-v-subset})
	\item Using the power to sample vertices uar from $A$ and $B$, we can estimate the {\em density} of $A$, namely $\frac{m(A,B)}{|A|}$ (and also density of $B$)
	in $\polylog n$ samples. This is done by sampling a random vertex $a \in A$, sampling a random neighbor $y\sim a$, and checking if $y\in B$ or not. (\Cref{lem:estrhos}) 
	\item Now, consider sampling $s$ uar vertices from each $A$ and $B$. Call the sampled sets $X \subseteq A$ and $Y \subseteq B$,
	and let $m(X,Y)$ be the number of edges from $X$ to $Y$. The expected value
	of $m(X,Y)$ can be shown to be $\frac{s^2\cdot m(A,B)}{|A||B|}$. And so if $m(X,Y)$ was concentrated around its mean, we
	would get our estimate of $m(A,B)$ by multiplying $\frac{s^2}{m(X,Y)}$ with (our estimates of) the densities of $A$ and $B$ from the previous bullet point, since
	\begin{equation}
		\frac{s^2}{m(X,Y)} \cdot \frac{m(A,B)}{|A|} \cdot \frac{m(A,B)}{|B|} \approx \frac{s^2}{\frac{s^2\cdot m(A,B)}{|A||B|}}
		\cdot \frac{m(A,B)}{|A|} \cdot \frac{m(A,B)}{|B|} ~~= m(A,B) \label{eq:idea} \notag
	\end{equation}
	By calculating the variance of the estimator, one can show that concentration occurs when $m(X,Y) = \Omega(1)$, or more precisely,
	when $s = \Omega(\sqrt{|A||B|/m(A,B)})$ (\Cref{lem:estimate-mab}). 
\end{asparaitem}
\smallskip

\noindent
Of course, we don't know $s$ a priori, but this is handled by starting with a small $s$ and doubling it till we 
see $m(X,Y) = \Omega(1)$. Let us perform a calculation relating $s$ and $p$. 
First, for any value of $p$, by our choice of heavy vertices, we expect $|H|$ to be closely concentrated around $np$.
So, $\max(|A|, |B|) \leq pn$. When $m(A,B) \approx m$ (ignoring polylog factors), if
every vertex in $A \cup B$ is sampled with probability $1/\sqrt{m}$, then the expected
number of edges between $X$ and $Y$ will be $\Omega(1)$. Thus, when $s$ reaches $\frac{np}{\sqrt{m}} = O(p\sqrt{n/d})$, 
we obtain a good estimate of $m$. \smallskip

\noindent
{\em Query Complexity.}
Let us analyze the number of queries required to compute these estimates. The main bottleneck
is sampling $X$ and $Y$, and then getting $m(X,Y)$; as argued above, all other estimates can be computed in $\otilde(1)$ queries.
In the standard model, $m(X,Y)$ can be computed by $s^2$ pair queries. Plugging the value of $s$ above, this means $\otilde(p^2n/d)$ queries.

The above was the complexity in the ``challenging'' case.
Now we balance with the work needed in the ``easy'' case, which recall was $O(1/p)$. 
The balancing point is when $p = (d/n)^{1/3}$ which gives us an $\Ot((n/d)^{1/3})$-query algorithm of \Thm{stan}
promised. Furthermore, in the advanced model where $m(X,Y)$ can be computed with query cost $s$ (instead of $s^2$). So,
our final query complexity would be balanced at $p = (d/n)^{1/4}$. This leads to the $\Ot((n/d)^{1/4})$-query algorithm of \Thm{advanced}.
Of course we can't set $p$ to be this value since we don't know $d$, but this is achieved by the standard ``guessing and doubling'' trick.
The full algorithm is given in~\Cref{alg:est-d-known-n-full-alg}. \medskip

%
\noindent
{\bf When $n$ is unknown.} 
%
%
To underscore our ideas, let us explain what ERS~\cite{ERS19} does and why it needs to know $n$.
The basic ERS estimator is simple and elegant: pick a uar vertex $v$ and a uar neighbor $u \sim v$. If $\deg(v) \leq \deg(u)$,
output $2\deg(v)$, else output $0$. This is an unbiased estimator for average degree whose variance-by-squared-mean is 
$O(\sqrt{n/d})$, and so these many samples suffice for convergence. Note we don't know $d$. If we do know $n$, 
then we can apply the guess-and-half trick which
starts with a guess $g = n$, takes $O(\sqrt{n/g})$ samples, and stops if the current estimate
is comparable to $g$, otherwise halves $g$ and moves on. Without knowledge of $n$, it is not clear how to even implement this algorithm.
Indeed, as mentioned earlier, we prove a lower bound showing that in the base model (w/o $\RE$), $\Omega(\sqrt{n})$ samples
are required to estimate $d$, regardless of $d$. \smallskip

\noindent
{\em Getting around with random edge samples (and no structure).} 
With $\RE$ queries, we circumvent this lower bound by computing a stopping criterion for the ERS algorithm. 
The algorithm is quite simple: choose $O(1/\eps)$ uar edges $(u,v)$ and set $\tau$ to be the maximum value of $\min(\deg(u), \deg(v))$.
Then repeatedly run the ERS unbiased estimator till you observe at least $O(1/\eps^2)$ vertices of degree $\geq \tau$.
We can show that the average value of estimator converges
(\Cref{lem:dest-conc}) in these many samples, and furthermore, the number of samples
is  $O(\sqrt{n/d})$ (\Cref{clm:run-approx-d}).
Balancing this out with the BT Harmonic estimator, we get a $O(n^{1/3})$ query algorithm. So, for the Base model
with $\RE$ queries, up to $\poly(\eps^{-1}\log n)$, there is no difference whether $n$ is known or not. \smallskip

\noindent
{\em Getting around with additive queries (and no random edge samples).} 
Suppose $\RE$ queries are not available but we have $\ADD$ queries. 
We show that with these strong structural queries, we can design an algorithm which returns either
an approximation to $n/d$, or an estimate to $n$. Note that in either case we can run the ERS algorithm; even a coarse estimate to $n/d$ determines
the number of times we run the ERS unbiased estimator for convergence.
Furthermore, 
in both cases, we will also get an approximation for $n$. So, with $\ADD$ queries, we can estimate both $d$ and $n$
in $O(\sqrt{n/d})$ queries, beating the birthday paradox $\Omega(\sqrt{n})$ barrier. 

The key to this algorithm is a combinatorial lemma (\Cref{lem:onehop}) asserting that one of the following three conditions must hold.
Either (i) $\sum_{v \in V} \deg^2(v) = O(m^{3/2})$, or (ii) there are $\Omega(\sqrt{m})$ vertices
with degree $\Omega(\sqrt{m})$, or (iii) the probability that a uar neighbor of a random vertex
has degree $\Omega(\sqrt{m})$ is $\Omega(\sqrt{d/n})$. In the latter two cases, it is easy
to design a $O(\sqrt{n/d})$ query algorithm that discovers a vertex $v$ with $\deg(v) = \Omega(\sqrt{m})$.
Since $\deg(v)$ is known, by estimating the fraction of vertices that neighbor $v$,
we can also get an estimate of $n$. Overall, in $O(\sqrt{n/d})$ queries, one can estimate $n$.

The first case is the interesting one. Let us do a rough calculation. 
Pick $s$ uar vertices. The expected number of edges among these vertices is $(s/n)^2m = s^2(d/n)$.
When $s \approx \sqrt{n/d}$, we expect this set to contain $\Theta(1)$ edges. We can now approximate
$\sqrt{n/d}$ by finding the smallest $s$ such that the randomly sampled set has $\Theta(1)$ edges.
The latter quantity can be determined by an $\ADD$ query of cost $s$. When $\sum_{v \in V} \deg^2(v) = O(m^{3/2})$,
we can prove this estimator has low variance. Overall, we can estimate $\sqrt{n/d}$ with
$O(\sqrt{n/d})$ queries. 

We prove that all our complexities are optimal, up to $\poly(\eps^{-1}\log n)$ dependencies. All of our lower bounds constructions
are based on the same template. The hard instances are disjoint collections of cliques, where the size parameters
are chosen differently depending on the setting. We discuss these in \Sec{lowerbounds}.

\subsection{Further Related work} \label{sec:related}

The area of sublinear algorithms for graph problems is a field in itself, and we refer the reader to Chap. 8-10 of Goldreich's book~\cite{G17-book} for an in-depth discussion of models
and property testing results. In this section, we focus only on results that are directly relevant. \smallskip

\noindent
{\em Network Analysis, Network Sampling and the The Full Neighborhood Access Model.} In the context of network analysis, the problem of estimating average degree was studied by Dasgupta, Kumar, and Sarl\'os~\cite{DKS14}.
To the best of our knowledge, they are the first to describe the full neighborhood access model. 
In line with practical applications, they do not assume
$\RV$ queries, but rather assume that the graph is rapidly mixing. This work spurred interest in the problem of actually implementing $\RV$ queries,
with assumption on the graph structure. Chierichetti, Dasgupta, Kumar, Lattanzi, and Sarl\'os~\cite{CDKLS16} design rejection sampling based algorithms
for generating uar vertices, and Chierichetti and Haddadan~\cite{ChHa18} give the optimal bound for this problem.
The network science and data mining literature also has 
focused on graph sampling algorithms, often referred to as the process of ``network sampling". The focus is often
on generating a ``representative sample" and the results are primarily empirical~\cite{MB11,ANK10,AhNeKo12,ANK14,SEGP16,SEGP17}.
Eden, Jain, Pinar, Ron, and Seshadhri~\cite{EdJa+18} use ideas for theoretical sublinear algorithms (like ERS) to practically estimate the degree distribution. 
Katzir, Liberty, and Somekh~\cite{KaLiSo11} give collision based estimators using random walks for approximating the number of vertices, when $\RV$ queries
are not available.
Size estimation using collision estimators is likely folklore, but a formal treatment is given by Ron and Tsur~\cite{RoTs16}, who
also study the problem of estimating the sizes of hidden sets in richer models.
\Tetek-Thorup~\cite{TT22} did a detailed study of the $\eps$ dependence of these algorithms, and proved the optimal bound (up to additive $\eps$
dependencies) of $\Theta(\eps^{-1}\sqrt{n/d})$ for the full neighborhood access model. They also studied alternate intermediate models
called the ``hash-ordered access" model. 
\smallskip

\noindent
{\em Proportional Sampling and Random Edge Samples.} 
Motivated by trying to understand the size of the web~\cite{HHMN00,BG08}, Motwani, Panigrahy, and Xu~\cite{MPX07} study the problem of estimating the average of an array using both uniform and proportional sampling.
When the array is the degree sequence of a graph, this problem is average degree estimation. 
In the graph context, the power of edge sampling was first shown by Aliakbarpour et al~\cite{ABGPRY18}.
They studied the problem of star counting and got improvements over previous algorithms of Gonen, Ron, and Shavitt~\cite{GRS11}.
Assadi, Kapralov, and Khanna show how random edge samples can be used for subgraph counting~\cite{AsKaKh19}. Eden and Rosenbaum~\cite{ER18} explicitly
call out the problem of estimating $m$ with access to random edge samples. Bishnu, Chanda, and Mishra~\cite{BCM25} showed 
improvements for triangle counting with edge samples.
\smallskip

\noindent
{\em The Additive Model and Global Queries.} There has been significant study on ``query models" to obtain information about a graph. These queries are closely
related to the $\ADD(S)$ queries of our advanced model. As mentioned earlier, we take inspiration
from the work of Grebinski-Kucherov~\cite{GK00} whose main objective was to learn a graph in as few queries. 
In~\cite{BeHa+20}, Beame, Har Peled, Ramamoorthy, Rashtchian, and Sinha study the problem 
of estimating $m$ using Bipartite Independent Set (BIS) and Independent Set (IS) queries. A BIS query takes input two sets $A$ and $B$
and returns $1$ if there is no edge with one endpoint in $A$ and the other in $B$, and $0$ otherwise. An IS query takes input a set $S$ and returns $1$ 
if the set $S$ is independent.~\cite{BeHa+20} show that in $O(\log^{14} n)$ BIS queries and $\min(\sqrt{m}, \frac{n^2}{m})$ IS queries, one can estimate $m$ (ignoring $\eps$-dependence).
These results were respectively improved to $O(\log^6 n)$ by Addanki, McGregor, and Musco~\cite{AdMcMu22} and  $\min(\sqrt{m}, \frac{n}{\sqrt{m}})$ by Chen, Levi, and Waingarten~\cite{ChLeWa20}.
We remind the reader that in all of the above, a query costs one unit irrespective of the size of the set(s) being queried, and these algorithms invoke their
oracles on linear-sized sets. Moreover, these algorithms require the vertex set
to be known in advance. So these results don't imply sublinear query algorithms in our setting.

Nevertheless, we believe that our results that use $\ADD$ queries, like \Thm{advanced}, can also be obtained (with a logarithmic overhead) with the weaker IS oracle model 
as well, even while paying cost equal to the size of the set passed to the oracle. This stems from two observations: one, we typically 
invoke $\ADD(S)$ on sets containing only $O(1)$ edges, and two, using an older result of Angluin and Chen~\cite{AnCh08}, one can 
{\em reconstruct} this very sparse graph $G[S]$ in $O(\log |S|)$ many IS queries. Since each IS query can cost at most $|S|$, 
we can simulate $\ADD(S)$ for sparse $S$'s in $\Ot(|S|)$ cost. We leave the details out for this extended abstract. \smallskip
%
%

\noindent
{\em Arboricity Connection.} There is a deep connection between degree estimation and the graph arboricity/degeneracy. This was first
observed by ERS~\cite{ERS19}, and exploited in further subgraph counting and edge sampling results~\cite{EdRo18,EdRoRo19,EdRoSe20}. \smallskip

%
%
%
%

\subsection{Explicit questions for future work} \label{sec:future}
%
There is a single source for the $\log n$ factors in
the query complexities of \Thm{stan} and \Thm{advanced}: the binning argument needed to construct quasi-regular graphs.
This does not seem inherent to the problem, but our algorithm and analysis heavily uses it.
We think that removing
these $\log n$ dependencies could lead to conceptually cleaner algorithms. 

All our algorithms have a quadratic dependence on $1/\eps$. 
It is not at all clear that this is necessary.  Indeed, \Tetek-Thorup~\cite{TT22}
did a detailed study of $\eps$ dependencies for estimating $d$ and show an $\Omega(\eps^{-1}\sqrt{n/d})$ lower
bound for the full neighborhood access model and give almost matching algorithms. BT~\cite{BT24} discover simpler algorithms in the process of studying $\eps$-dependencies for estimating the average
of an array. It would be of great interest to pin down these dependencies in \Cref{tab:estd-vertical}. For example, in the unknown $n$
setting with $\RE$, could the $\eps$ dependencies create a separation between the models?

\begin{question}
	Can one improve the dependence on $\poly(\log n, \frac{1}{\eps})$ factors in our algorithms?
\end{question}
\noindent
As noted above, many sublinear graph results have shown that the arboricity
can be used to get better query complexities. The ERS analysis actually proves a bound of $O(\alpha/d)$,
where $\alpha$ is the graph arboricity/degeneracy. In the worst case, $\alpha \leq \sqrt{2m}$, giving
the $O(\sqrt{n/d})$ bound we have seen. This leads to the following question.

\begin{question}
Can every occurrence of $\sqrt{n/d}$ in our algorithms be replaced by the (potentially much smaller) quantity $\alpha/d$?
\end{question}

\noindent
It is likely that $\ADD$ queries can
help in other sublinear graph algorithms. As a first step, we can investigate whether they help
in triangle counting. \Tetek-Thorup~\cite{TT22} have given improvements for triangle counting in the full neighborhood access model;
could their algorithm be implemented with the weaker $\ADD$ queries?

\section{Preliminaries}

Let $G=(V,E)$ be a simple undirected graph. We use the following notation.
\begin{itemize}[noitemsep]
	\item For $x\in V$, $\deg(x)$ is the number of neighbors of $x$ in $G$.
	\item For $S\subseteq V$, $\deg(S) := \sum_{x\in S} \deg(x)$.
	\item Given $S\subseteq V$, $E(S) \subseteq E$ denotes the edges $e\in E$ with $e\cap S\neq \emptyset$. We let $m(S) := |E(S)|$.
	\item Given $A, B$ subsets of $V$, $E(A,B)$ denotes the edges $e\in E$ with one endpoint in $A$ and the other in $B$. We let $m(A,B) = |E(A,B)|$.
    \item For edge $e = (u,v) \in E$, the degree of $e$ is $\min(\deg(u), \deg(v))$.
\end{itemize}
\noindent
Recall the graph access models from~\Cref{sec:results}. We call an implicitly defined subset $A\subseteq V$ {\em query-able} if given any vertex $x \in V$, we can check whether $x\in A$
in $O(1)$ queries. Most commonly, $A$ will be a subset of vertices whose degrees are in a specific range and so one can query membership using $\DEG$ queries. \smallskip

\noindent
Throughout we use ``whp'' to mean with probability $\geq 0.99$. One could replace this by $1-\delta$ while paying a $O(\ln(1/\delta))$ factor in the query complexity 
by using standard ``median-of-means'' boosting. \smallskip

\noindent
We use notation $\Pr$, $\Exp$, and $\Var$ to denote ``probability'', ``expectation'', and ``variance''. We will continually use the following basic concentration facts.

\begin{fact}[Markov and Chebyshev]\label{fact:mkov-cheb}~
	
	\noindent
	For any non-negative random variable $Z$, $\Pr[Z \geq t] \leq \frac{\Exp[Z]}{t}$.
	
	\noindent
	For any random variable $Z$, $\Pr[|Z - \Exp[Z]| \geq t] \leq \frac{\Var[Z]}{t^2}$. In particular, for any $\eps > 0$, $\Pr[Z \notin (1\pm \eps) \Exp[Z]] \leq \frac{\Var[Z]}{\eps^2(\Exp[Z])^2}$.
\end{fact}

\begin{fact}[Chernoff Bounds]\label{fact:chern}~
	
	\noindent
	Let $X_1, \ldots, X_t$ be Bernoulli random variables with $q \leq \Pr[X_i = 1] \leq p$.
	Then,
	\begin{enumerate}[label=(\alph*),noitemsep]
		\item $\Pr[\sum_{t=1}^t X_i \leq (1-\eps)qt] \leq \exp(-\eps^2 qt/2)$ for $\eps \in (0,1)$.
		\item $\Pr[\sum_{t=1}^t X_i \geq (1+\eps)pt] \leq \exp(-\eps^2 pt/4)$ for $\eps \in (0,4)$.
		\item $\Pr[\sum_{i=1}^t X_i \geq \frac{t}{2}] \leq \exp(-\frac{t}{4}\ln(1/2p))$ when $p\leq 1/16$.
	\end{enumerate}
	 
\end{fact}

\subsection{Subroutines}	
\def\Samp{\mathbf{Samp}}
We use the standard estimator for the bias of a Bernoulli random variable. This was studied by~\cite{LNSS93} and the following description is from~\cite{Wata05}.
The setting is the following: $\Samp()$ returns $1$ with probability $p$ and $0$ otherwise, and the objective is to estimate $p$.

\newcommand{\BernoulliEst}{{\sc BernoulliEst}}
\begin{algorithm}
	\caption{{\sc \BernoulliEst($\eps \in (0,1)$, $\delta \in (0,1)$)}\cite{Wata05}}\label{alg:ber}
	\begin{algorithmic}[1]
		\LineComment{Access to $\Samp()$ that returns $1$ with probability $p$ and $0$ otherwise.}
		\LineComment{Returns: $\widehat{p}$ estimate of $p$}
		\LineComment{Query Complexity: $O(\frac{1}{p\eps^2}\ln(2/\delta))$}
		\Statex
		\State $M \eq 0$; $N\eq 0$
		\While{$M < \frac{3(1+\eps)}{\eps^2}\ln(2/\delta)$}:
			\State $M \eq M + \Samp()$; $N\eq N+1$
		\EndWhile
		\State \Return $\widehat{p} \eq M/N$.
	\end{algorithmic}
\end{algorithm}
\begin{lemma}\label{lem:ber}\cite{LNSS93,Wata05}
	\BernoulliEst$(\eps,\delta)$ returns $\widehat{p}$ in $O(\frac{1}{\eps^2 p}\ln(1/\delta))$  queries and $\Pr[|\widehat{p} - p| \geq \eps p] \leq \delta$.
\end{lemma}
\noindent
We use the above estimator for three different fractions, and we describe these below.

\newcommand{\EstFracVert}{{\sc EstFracVert~}}
\newcommand{\EstFracEdge}{{\sc EstFracEdge~}}
\newcommand{\EstNumVer}{{\sc EstNumVer~}}
\begin{lemma}\label{lem:ber-subs}
	The following randomized algorithms exist
\begin{itemize}[noitemsep]
	\item \EstFracVert takes input a query-able set $A$ and whp returns an $(1\pm \eps)$ estimate of $\frac{|A|}{n}$ in $O(\frac{n}{|A|\eps^2})$ queries.
	\item \EstFracEdge takes input two query-able sets $A,B$ and whp  returns an $(1\pm \eps)$ estimate of $\frac{m(A,B)}{m}$ in $O(\frac{m}{\eps^2 m(A,B)})$ queries.
	\item \EstNumVer takes input a vertex $x\in V$, and whp returns an $(1\pm \eps)$ estimate of $n$ in $O(\frac{n}{\eps^2 \deg(x)})$ queries.
\end{itemize}
\end{lemma}
\begin{proof}
	We run \BernoulliEst~ with the following $\Samp()$s using the oracles above.
	\begin{itemize}[noitemsep]
		\item Run $\RV$ to get $x\in V$ and return $1$ if $x\in A$.
		\item Run $\RE$ to get $e\in E$ and return $1$ if $e\in E(A,B)$.
		\item Run $\RV$ to get $y$ and return $1$
		 if $y\sim x$ using $\PAIR$ (or $\ADD$)\qedhere
	\end{itemize}
\end{proof}

\noindent
Beretta and \Tetek~\cite{BT24} give a $O(\eps^{-2}\davg)$ query $(1\pm \eps)$-approximation to the average degree.
Their algorithm only needs $\RE$ and $\DEG$ oracles, and does not need to know $n$.
For completeness we provide a full description and analysis in \Cref{sec:app:bt}. 

\begin{restatable}{theorem}{bertet}\label{lem:bt}
	[Theorem 5.1 of Beretta-\Tetek~\cite{BT24}, paraphrased]~
	
	\noindent
	Given access to an unknown graph via $\RE$ and $\DEG$ oracles, there is a randomized algorithm {\sc \BTEstAvgDeg}$(\eps)$ that whp
	returns an $(1\pm \eps)$-approximation to $\davg$ in $O(\eps^{-2}\davg)$ queries.
\end{restatable}


\noindent
Eden, Ron, and Seshadhri~\cite{ERS19} give an algorithm with access to $\RV, \RN$ and $\DEG$ oracles (basic model) and {\em with} knowledge of $n$ returns 
an $(1\pm \eps)$-approximation to $\davg$. We provide a proof in~\Cref{sec:app:ers} for completeness.

\begin{restatable}{theorem}{ers}\label{thm:ers}
[Theorem 4 in~\cite{ERS19}, paraphrased with $s=1$; Theorem 13 of \cite{Sesh15}, paraphrased]
	Given access to an unknown graph via the base model (that is, with $\RV$, $\RN$, and $\DEG$), and the knowledge of $n$, there is a randomized algorithm {\sc \ERSEstAvgDeg}$(\eps)$ which whp
	returns an $(1\pm \eps)$-approximation to $\davg$ in $O(\eps^{-2}\sqrt{n/\davg})$ queries.
\end{restatable}
%

\noindent
When the size of a universe is unknown but we can obtain uar samples from it, then we can estimate the size in roughly square-root of the universe size random samples
by collision counting. One place where this is formalized is the following theorem of Ron and Tsur~\cite{RoTs16}. We provide a proof in \Cref{app:rt} for completeness.

\begin{restatable}{theorem}{rt}\label{thm:rt}
	[Theorem 2.1 of \cite{RoTs16}, paraphrased]
	Given access to an unknown graph via just $\RV$ samples, there is a randomized algorithm {\sc EstNumColl}($\eps$) which
	returns an $(1\pm \eps)$-approximation $\widehat{n}$ to $n$ in $O(\eps^{-2}\sqrt{n})$ samples.
\end{restatable}



\section{Estimating average degree with known number of vertices}


\noindent
The first step of the algorithm estimates the fraction of edges that are incident to ``very high'' degree vertices. 
We sample $\approx \log n/p$ vertices uar and use $\Delta$ to denote the $(\log n)$th largest degree.
Here $p$ is a parameter; think of $p\approx n^{-1/4}$.
So, we expect $\approx n^{3/4}$ to have degree larger than $\Delta$. We then use random edges to estimate the fraction of {\em edges} incident to these high-degree vertices.
If this fraction is bounded away from $1$, then one can use the {\sc HarmonicEstimator} of~\cite{BT24} to estimate the average degree.
The bulk of the work will be taking care of the other case and this is described in~\Cref{sec:3.1}.

%
%
%
%

\begin{algorithm}
	\caption{{\sc \AlgEstDensHighDeg($p \in (0,1)$)}}\label{alg:esthighdeg}
	\begin{algorithmic}[1]
		\LineComment{Assume access to $\RV$, $\RE$, and $\DEG$.}
		\LineComment{$p \in (0,1)$ is a parameter}
		\LineComment{Outputs $(f,\Delta)$; $\Delta \approx $ $(pn)$th largest vertex degree \& $f \approx $ fraction of edges with degree $\geq \Delta$}
		\LineComment{Time: $O(\log n/p)$}
		\Statex
		
	\State Sample $r = \ceil{\frac{\log n}{p}}$ vertices uar and obtain their degrees. \Comment{Uses $\RV$ and $\DEG$}
	\State Let $\Delta$ be the $\ceil{\log n}$th largest degree in this sample.
	\State Sample $t = \ceil{\log n}$ edges uar and query their degrees. \Comment{Uses $\RE$ and $\DEG$}
	\State Let $f$ (respectively $f'$) be the fraction of these edges with degree $\geq \Delta$ (respectively $\geq \Delta+1$).
	\State \Return $(f, f', \Delta)$. \Comment{$f'$ is calculated for a technical reason mentioned below}
	\end{algorithmic}
\end{algorithm}
%
%

\begin{claim}\label{clm:chern}
	Let $(f,f',\Delta)$ be the output of \hyperref[alg:esthighdeg]\AlgEstDensHighDeg\emph{($p$)}. Let $H := \{x\in V~:~\deg(x)\geq \Delta\}$
	and $H' := \{x\in V~:~\deg(x)\geq \Delta+1\}$. 
	Then, 
	\begin{enumerate}[label=(\alph*),noitemsep]
	\item $\Pr[|H| \leq \frac{np}{4}] = o(1)$.
	\item $\Pr[|H'| \geq 2np] = o(1)$.
	\item If $f < \frac{2}{3}$, then $\Pr[|E[H]| \geq \frac{3m}{4}] = o(1)$. An analogous statement with $f', H'$ holds.
	\item If $f \geq \frac{2}{3}$, then $\Pr[|E[H]| < \frac{m}{2}] = o(1)$. An analogous statement with $f', H'$ holds.
	\end{enumerate}
\end{claim}
\begin{proof}
	Let $R$ be the subset of randomly sampled vertices. Let $A$ be the set of the largest $np/4$ vertices (breaking ties arbitrarily).
	If $R$ contains  $< \log n$ vertices from $A$, then note that $A\subseteq H$. So, $\Pr[|H|\leq \frac{np}{4}] \leq \Pr[|R\cap A| \geq \log n]$.
	Since $\Exp[|R\cap A|] = \frac{\log n}{4}$, this probability is $o(1)$ by Chernoff bounds (\Cref{fact:chern}(b)).
	Similarly, if we let $B$ be the set of largest $2np$ vertices, and if $R$ contains $> \log n$ vertices, then $H' \subsetneq B$. 
	So, $\Pr[|H'| \geq 2np] \leq \Pr[|R\cap B| \leq \log n]$. Since $\Exp[|R\cap B|] = 2\log n$, this probability is $o(1)$ by Chernoff bounds (\Cref{fact:chern}(b)).
	The third and fourth bullet points also hold due to Chernoff bounds (\Cref{fact:chern}(a),(b)) more straightforwardly.
\end{proof}

\noindent
\begin{claim}\label{clm:obs1}
	Let $\Delta, H'$ be as in~\Cref{clm:chern}. If $|E[H']| < \frac{3m}{4}$, then $\davg \leq 8\Delta$. 
\end{claim}
\begin{proof}
	Note that $\sum_{x\notin H'}\deg(x) \geq |E\setminus E[H']|$ and since $x\notin H'$ implies $\deg(x) \leq \Delta$, we get $n\Delta \geq m/4$ and the claim follows by rearranging
	and recalling $\davg = 2m/n$.
\end{proof}

\noindent
The above claims imply if $f' < \frac{2}{3}$, then whp we have an upper bound of $O(\Delta)$ on $\davg$. Furthermore, the set $H$ of vertices with degree $\geq \Delta$ is ``large'', that is, 
of size $\Omega(np)$. This is where the harmonic estimator of~\cite{BT24} is applicable: (i) since $H$ is large one can estimates its size in $\approx 1/p$ queries using \EstFracVert, 
and (ii) sampling vertices proportional to degree (which is possible due to edge sampling) and only considering those in $H$ with weight inversely proportional to the probability they are sampled with (i.e. their degree) gives an estimate of $|H|/m$. The above two together give the desired estimate for $\davg$. The following pseudocode is from~\cite{BT24} slightly modified
to our graph setting; we provide it just for completeness.

%
%

\begin{algorithm}
	\caption{{\sc HarmonicEstimator($\Delta, p, \eps \in (0,1)$)}}\label{alg:harmest}
	\begin{algorithmic}[1]
			\LineComment{Assume access to $\RV$, $\RE$, and $\DEG$.}
				\LineComment{Assume: $(f, f', \Delta)$ is the output of \hyperref[alg:esthighdeg]\AlgEstDensHighDeg($p$) {\bf and} $f' < \frac{2}{3}$}
				\LineComment{Returns estimate $\widehat{\davg}$ of average degree $\davg$ of $G$}
				\LineComment{Query Complexity: $O\left(\frac{1}{\eps^2p}\right)$ }
				\Statex
		\State Define $H = \{x:\deg(x) \geq \Delta\}$; we have query access to $H$.
		\Comment{$|H| = \Omega(np)$ and $\davg \leq 8\Delta$}
		\State $\widehat{p} \eq $ \hyperref[lem:ber-subs]{\EstFracVert}($H, \eps$) \Comment{Query Complexity: $O\left(\frac{1}{\eps^2p}\right)$} \Comment{Uses $\RV$}
		\State $k\eq \ceil{\frac{1000}{\eps^2 \widehat{p}}}$ \Comment{The constant $1000$ is arbitrary and could potentially be made smaller}
		\For{$i=1$ to $k$}:
		\State Sample an edge $e\in E$ uar and sample $x\in e$ uar. \Comment{Uses $\RE$} \label{alg:he:line-samp}
		\If{$\deg(x) \geq \Delta$}: \Comment{Uses $\DEG$}
		\State $Z_i \eq 1/\deg(x)$
		\Else:
		\State $Z_i \eq 0$
		\EndIf
		\EndFor
		\State $Z \eq \frac{1}{k} \cdot \sum_{i=1}^k b_i$
		\State \Return $\widehat{\davg} \eq \widehat{p}/Z$
	\end{algorithmic}
\end{algorithm}

\begin{lemma}[Paraphrasing Lemma 3.1 of \cite{BT24}] \label{lem:harmest}
	
	\noindent
	Let $(f,f',\Delta)$ be the output of {\AlgEstDensHighDeg}($p$) and suppose $f' < \frac{2}{3}$. Then, (i) the expected query complexity  
	of {\sc HarmonicEstimator} is $O(\frac{1}{\eps^2 p})$, and 
	(ii) with probability $> 9/10$, the output $\widehat{\davg}$ of {\sc HarmonicEstimator} satisfies $\widehat{\davg} \in (1\pm \eps)\davg$.
\end{lemma}
\begin{proof}
	The probability a vertex $x$ is sampled in~\Cref{alg:he:line-samp} is $\deg(x)/2m$.
	Note that $\Exp[Z_i] = \sum_{x\in H} \frac{\deg(x)}{2m}\cdot \frac{1}{\deg(x)} = \frac{|H|}{2m}$. 
	Using the fact that $\davg = O(\Delta)$, it can be shown (see~\cite{BT24}) that $\frac{\Var[Z_i]}{\Exp[Z_i]^2} = O(1/p)$, 
	and this means $Z$ is a good estimate to $\frac{|H|}{2m}$. Since $\widehat{p}$ is a good estimate to $\frac{|H|}{n}$, we get that
	$\widehat{\davg} = \widehat{p}/Z$ is a good estimate to $\frac{2m}{n} = \davg$.
\end{proof}
\subsection{The case when high-degree vertices induce a majority of the edges}\label{sec:3.1}

In this section we assume that $(f,f',\Delta)$ returned by \Cref{alg:esthighdeg} has $f'\geq 2/3$.
So, we know  that whp, the set $H' := \{x~:~\deg(x) \geq \Delta+1\}$ has size $|H'| = O(np)$,
and, since $f' \geq 2/3$, we have $|E[H']| \geq m/2$. A majority of the edges are induced by $H'$, which is a query-able subset.
The algorithm proceeds as follows.
\begin{enumerate}[noitemsep]
	\item It finds quasi-regular, query-able subsets $H'_i, H'_j$ of $H'$ such that $|E(H'_i:H'_j)| = \Omega(m/\log^2 n)$.
	\item It estimates $m_{ij} := |E(H'_i:H'_j)|$.
    \item Using $m_{ij}$, it estimates $m$, and thereby, estimates $\davg$.
\end{enumerate}

\begin{definition}\label{def:regular}
	A subset $A\subseteq V$ is called {\em quasi-$k$-regular} if for some parameter $k \in \NN$,  $\deg(x) \in [k, 2k]$ for all $x\in A$. 
\end{definition}

\noindent
The set $H'$ can be partitioned into quasi-regular sets $H'_1\cup H'_2 \cup \cdots \cup H'_L$ where
\[
	H_i := \{x\in V~:~ 2^{i-1} (\Delta+1) \leq \deg(x) \leq 2^{i} (\Delta+1) \} ~~~\text{and}~~~ L = \ceil{\log_2(n/(\Delta+1))}
\]
Note that given any vertex $x\in V$, its degree tells us whether $x\in H'_i$. So every $H'_i$ is query-able.

\begin{definition}\label{def:dense}
	A pair of subsets $A, B$ is called $\alpha$-dense, if $|E(A,B)| \geq m/\alpha$.
	$A$ and $B$ will either be disjoint or the same set, in which case $E(A,A) = E[A]$.
\end{definition}

\begin{lemma}\label{lem:recogij}
	There exists $i,j \in [L]$ such that $H'_i, H'_j$ are $\log^2 n$-dense. Furthermore, there is a randomized
	algorithm using $\RE$ and $\DEG$ oracles that runs in $O(\log^3 n)$ queries and finds an $i$ and $j$ such that with $1-o(1)$ probability $H'_i, H'_j$ are $\frac{\log^2 n}{16}$-dense. 
\end{lemma}
\begin{proof}
    The existence of $i$ and $j$ simply follows from $|E(H',H')| \geq m/2$ and the fact there are $\leq \frac{\log^2 n}{2}$-pairs.
	Consider an that algorithm samples $\ceil{\log^3 n}$ edges uniformly at random and sets
	$n_{ij}$ to be the number of such edges with one endpoint in $H'_i$ and the other in $H'_j$ (recall, these sets are query-able).
	It returns any $(i,j)$ with $n_{ij} \geq \frac{\log^2 n}{4}$. If $(i,j)$ were indeed $\log^2 n$-dense, then $\Exp[n_{ij}] \geq \log^2 n$, 
	and so, by Chernoff bounds, with $1-o(1)$ probability $n_{ij} \geq \log^2 n/4$. Similarly, if $E(H'_i,H'_j) < \frac{m}{16\log^2 n}$, 
	we would have $\Exp[n_{ij}] < \frac{\log^2 n}{8}$, and with $1-o(1)$ probability we won't output such a pair.
\end{proof}
\noindent
Henceforth, assume we know $i$ and $j$ such that $(H'_i, H'_j)$ form a $O(\log^2 n)$-dense pair of quasi-regular sets. 
In the next few lemmas, for simplicity, we consider arbitrary regular dense sets $(A,B)$.
At the very end, we will apply these lemmas on $H'_i$ and $H'_j$.

\begin{lemma}\label{lem:rand-v-subset}
	Let $A\subseteq V$ be a query-able quasi-$k$-regular subset with $|E(A)| \geq m/\alpha$. There is a randomized algorithm
	that returns a uar $x\in A$ 
	in $O(\alpha)$ expected queries.
\end{lemma}
\begin{proof}
	This is by rejection sampling; we state the algorithm below.
	\begin{algorithm}
		\caption{{\sc RandVertSamp($A, k$)}}\label{alg:rv-subset}
		\begin{algorithmic}[1]
			\LineComment{Assume access to $\RE$, and $\DEG$.}
			\LineComment{Assume $A\subseteq V$ is a query-able quasi-$k$-regular subset with $|E(A)| \geq m/\alpha$}
			\LineComment{Returns $a\in A$ uar in expected $O(\alpha)$ queries.}
			\Statex
			\State Sample an edge $e$ uar; {\bf reject} and start over if $e\notin E(A)$. \Comment{Uses $\RE$}
			\State Sample $x\in e$ uar; {\bf reject} and start over  if $x\notin A$. \Comment{Uses $\DEG$}
			\State Return $x$ with probability $\frac{k}{\deg(x)}$; {\bf reject} and start over otherwise.
		\end{algorithmic}
	\end{algorithm}
	
	\noindent
	By design only $x\in A$ is returned and the probability a particular $x\in A$ is returned is precisely
	$\sum_{e\sim x} \Pr[e ~\text{sampled}] \cdot \Pr[x~\text{sampled}~|~e~\text{sampled}~]\cdot \frac{k}{\deg(x)} = \frac{k}{2m}$, which is independent of $x$.
	Thus, conditioned on not being rejected, the probability is uniform over $A$. The probability of
	not being rejected is $k|A|/2m$ and since $2k|A| \geq \sum_{x\in A}\deg(x) \ge |E(A)| \geq m/\alpha$, we get that the probability of success is $\geq \frac{1}{4\alpha}$.
\end{proof}
\noindent
If $(A,B)$ is a $\alpha$-dense pair of quasi-regular sets, then $E(A,B) \geq m/\alpha$ implies $E(A) \geq m/\alpha$. The above lemma gives us
the power to uniformly sample from $A$ and $B$. We can now estimate the ``densities'' of $A$ and $B$.

\begin{lemma}\label{lem:estrhos}
	Let $A, B\subseteq V$ be an $\alpha$-dense pair of subsets where $A$ is quasi-$k$-regular and $B$ is quasi-$\ell$-regular. 
	Let $\rho_A := \frac{m(A,B)}{|A|}$ and $\rho_B := \frac{m(A,B)}{|B|}$ 
	be the ``densities'' of $A$ and $B$. There exists a randomized algorithm which with $> 23/25$ probability returns an $(1\pm \eps)$-estimate of $\rho_A$ and $\rho_B$
	and runs in expected $O(\alpha^2/\eps^2)$ queries.
\end{lemma}
\begin{proof}
	We describe the algorithm for $\rho_A$ and the algorithm for $\rho_B$ is symmetric. 
The idea is to simply sample a random vertex from $A$ and then use it to estimate the number of neighbors in $B$.

	\begin{algorithm}[h!]
	\caption{{\sc EstimateDensity($A, B, k, \ell, \eps \in (0,1)$)}}\label{alg:est-dens}
	\begin{algorithmic}[1]
		\LineComment{Assume access to $\RE$, $\RN$, and $\DEG$.}
		\LineComment{$(A,B)$ is assumed to be $\alpha$-dense; $A$ is quasi-$k$-regular and $B$ is quasi-$\ell$-regular.}
		\LineComment{Returns $\widehat{\rho}_A$, $\widehat{\rho}_B$ estimates of $\frac{m(A,B)}{|A|}$ and $\frac{m(A,B)}{|B|}$}
		\LineComment{Query Complexity: $O(\alpha^2/\eps^2)$}
		\Statex
		\State $t \eq \ceil{\frac{100\alpha}{\eps^2}}$
		\For{$i\in \{1,2,\ldots, t\}$}:
			\State Sample $a\in A$ uar using \hyperref[alg:rv-subset]{{\sc RandVertSamp}($A,k$)}. \label{alg:estdense:samp} \Comment{Takes $O(\alpha)$ queries}
			\State Sample $y \sim a$ uar. \Comment{Uses $\RN$}
			\If{$y\in B$}:
				\State $Z_i \eq \deg(a)$
			\Else:
				\State $Z_i \eq 0$
			\EndIf
		\EndFor
		\State \Return $\widehat{\rho}_A \eq \frac{1}{t} \sum_{i=1}^t Z_i$ \Comment{The algorithm also returns $\widehat{\rho_B}$ in an analogous manner}
	\end{algorithmic}
\end{algorithm}

\noindent
For $a\in A$, let $\deg(a, B)$ denote the number of neighbors of $a$ in $B$. Note that $\sum_{a\in A} \deg(a,B) = m(A,B)$.
Note that for every $i\in [k]$, $Z_i$ is an unbiased estimate of $\rho_A$,
\[
\Exp[Z_i] = \sum_{a\in A} \frac{1}{|A|} \cdot \frac{\deg(a,B)}{\deg(a)} \cdot \deg(a) = \frac{m(A,B)}{|A|} = \rho_A
\]
with variance
\[
\Var[Z_i] \leq \Exp[Z^2_i] = \sum_{a\in A} \frac{1}{|A|} \cdot \frac{\deg(a,B)}{\deg(a)} \cdot \deg^2(a) = \frac{1}{|A|}\sum_{a\in A} \deg(a,B)\deg(a)	
\]
Since $\deg(a) \leq 2k$, we get the variance can be upper-bounded by
$
	\Var[Z_i] \leq \frac{2k}{|A|}\cdot \sum_{a\in A} \deg(a,B) = 2k\rho_A$.
Furthermore, $k|A| \leq \sum_{a\in A} \deg(a) \leq 2m$ implies $k\leq 2m/|A|$. Finally, using the fact that $m(A,B) \geq m/\alpha$, 
we get that $\rho_A \geq \frac{m}{\alpha|A|}$. All this can be plugged above to give $\Var[Z_i] \leq 4\alpha \rho^2_A$.
Since we take a average of $\geq \frac{100\alpha}{\eps^2}$ samples, the variance of $\widehat{\rho}_A$ is $\leq \frac{\eps^2\rho^2_A}{25}$ while $\Exp[\widehat{\rho}_A] =\rho_A$.
Chebyshev now gives that $\Pr[\widehat{\rho_A} \notin (1\pm \eps)\rho_A] \leq \frac{1}{25}$.
The query complexity of the algorithm in each for-loop is dominated by~\Cref{alg:estdense:samp} which takes $O(\alpha)$ queries by~\Cref{lem:rand-v-subset}. Since there are $O(\alpha/\eps^2)$ loops, the lemma follows.
\end{proof}

\def\abort{\mathbf{abort}}


The most challenging part of the algorithm and analysis are in the next lemma.

\begin{lemma}\label{lem:estimate-mab}
	Let $A, B\subseteq V$ be an $\alpha$-dense pair of subsets and say $A$ is quasi-$k$-regular and $B$ is quasi-$\ell$-regular. Let $\mu := \frac{m(A,B)}{|A||B|}$.
	There is a randomized algorithm with $\RV$, $\RN$, $\DEG$, $\RE$ and either $\ADD$ or the weaker $\PAIR$ oracle access,  
	which takes $\eps \in (0,1)$ and $s\in \NN$ as parameters and either outputs $\abort$ or $\widehat{m}(A,B)$ with the following guarantees.
	\begin{enumerate}[label=(\alph*), noitemsep]
		\item The probability of $\abort$ is at least $1 - s^2\mu$.
		\item If $s^2\mu \!\geq\!\!0.1$ and algorithm doesn't $\abort$, then wp $>\!\frac{4}{5}$ it returns $\widehat{m}(A,B)\!\! \in \!\!(1\pm 3\eps)m(A,B)$.
		\item If $s^2\mu \geq 2$, the probability algorithm returns $\abort$ is $< \frac{C\alpha}{s}$.
	\end{enumerate}
\noindent
	The query complexity of the algorithm is $O(\alpha^2s/\eps^2)$ with $\ADD$ and $O(\alpha^2s^2/\eps^2)$ with $\PAIR$.
\end{lemma}

%
%
\begin{proof}[Proof of~\Cref{lem:estimate-mab}.]
The algorithm samples $s$ vertices from both $A$ and $B$ and uses the $\ADD$ query to find the number of edges between these subsets in $O(s)$ queries.
If it finds no edges, then it aborts. Otherwise it takes an average of $O(\alpha/\eps^2)$ such samples to reduce variance.

\begin{algorithm}[ht!]
	\caption{{\sc \EstNumEdges($A, B, k , \ell, s\in \NN, \eps \in (0,1)$)}}\label{alg:est-mab}
	\begin{algorithmic}[1]
		\LineComment{Assume access to $\RN$, $\RE$, $\DEG$ and either $\ADD$ or weaker $\PAIR$.}
		\LineComment{$A,B$ is an $\alpha$-dense pair of quasi-regular sets.}
		\LineComment{The algorithm either outputs $\abort$ or an estimate $\widehat{m}(A,B)$ of $m(A,B)$.}
		\LineComment{Query Complexity: $O(\alpha^2 s/\eps^2)$ with $\ADD$ and $O(\alpha^2s^2/\eps^2)$ with $\PAIR$.}	
		\Statex
		
		\State Sample uniform $X\subseteq A$ and $Y\subseteq B$ with $|X| = |Y| = s$. 
		\Comment{Use $O(s)$ runs of~\hyperref[alg:rv-subset]{\sc RandVertSamp} on $(A,k)$ and $(B,\ell)$. Query Complexity: $O(s\alpha)$.}
		\State Query $m(X,Y)$ using $\ADD$ or $\PAIR$ and $\abort$ if answer is $0$. Otherwise proceed. \LineComment{Query Complexity of previous line: $O(s)$ with $\ADD$ and $O(s^2)$ with $\PAIR$.}
		\LineComment{Otherwise, repeat the above $O(\alpha/\eps^2)$ times and average.}
		\State $t \eq \ceil{\frac{180\alpha}{\eps^2}}$
		\For{$i\in \{1,2,\ldots, t\}$}:
		\State Sample uniform $X\subseteq A$ and $Y\subseteq B$ with $|X| = |Y| = s$.  \Comment{Query Complexity: $O(\alpha s)$}
	
		\State Set $Z^{(i)} \eq m(X,Y)$ \Comment{Use $\ADD$ in $O(s)$ queries or $\PAIR$ in $O(s^2)$ queries.}
		\EndFor
		\State Set $Z_s \eq \frac{1}{t}\sum_{i=1}^t Z^{(i)}$
		\State $(\widehat{\rho}_A, \widehat{\rho}_B) \eq $ \hyperref[alg:est-dens]{{\sc EstimateDensity}}($A,B,k,\ell,\eps$) \Comment{Expected queries: $O(\alpha^2/\eps^2)$.}
		\State \Return $\widehat{m}(A,B) \eq \frac{s^2\widehat{\rho}_A\widehat{\rho}_B}{Z_s}$
	\end{algorithmic}
\end{algorithm}
\noindent
The expected query complexity can be seen to be $O(\alpha^2 s/\eps^2)$ with $\ADD$ queries
    and $O(\alpha^2s^2/\eps^2)$ with $\PAIR$ queries.  We now establish parts (a), (b), and (c).
	For $s\in \NN$ and let $Z = m(X,Y)$ where $(X,Y)$ are as described above.
%
	\begin{claim}\label{clm:exp}
		$\Exp[Z] = \frac{s^2}{|A||B|}\cdot m(A,B) = s^2\mu = \frac{s^2 \rho_A \rho_B}{m(A,B)}$
	\end{claim}
	\begin{proof}
		For every $e\in E(A,B)$ let $X_e$ denote the indicator that both endpoints are sampled.
		The proof follows by noting that $\Exp[Z] = \Exp[\sum_{e\in E(A,B)} X_e]$ and $\Exp[X_e] = \frac{s}{|A|}\cdot \frac{s}{|B|}$.
	\end{proof}
	
	\noindent
	Markov's inequality (\Cref{fact:mkov-cheb}) implies (a) of the lemma. We now establish (b) and (c) via a variance calculation.
	
	\begin{claim}\label{clm:var-calc}
		$
			\Var[Z] < s^2\mu (1 - s^2\mu/2) ~	+ \frac{8\alpha}{s} \cdot \left(s^2\mu\right)^2
		$
	\end{claim}
	\begin{proof}
	We first establish that
	\begin{equation}\label{eq:var}
		\Var[Z] < \Exp[Z](1 - \Exp[Z]/2) ~+~ \frac{s^3}{2|A||B|}\cdot \left(\sum_{a\in A} \frac{\deg^2(a,B)}{|B|} + \sum_{b\in B} \frac{\deg^2(b,A)}{|A|}\right) 
	\end{equation}
	\noindent
		As is standard, we start with \[\Exp[Z^2] = \Exp\left[\left(\sum_{e\in E(A,B)} X_e\right)^2\right] = \underbrace{\sum_{e\in E(A,B)} \Exp[X_e]}_{=\Exp[Z_s]} + \sum_{e\neq f \in E(A,B)} \Exp[X_eX_f]\]
		Note that if $e$ and $f$ don't share endpoints then $\Exp[X_eX_f] = \left(\frac{s^2}{|A||B|}\right)^2$. If they do share an endpoint $a\in A$, 
		then $\Exp[X_eX_f] = \frac{s^3}{|A|^2|B|}$, and if they share an endpoint $b\in B$, then $\Exp[X_eX_f] = \frac{s^3}{|A||B|^2}$.
		Crudely, $\sum_{e,f} \left(\frac{s^2}{|A||B|}\right)^2 \leq \binom{m(A,B)}{2} \cdot \left(\frac{s^2}{|A||B|}\right)^2 < \frac{(\Exp[Z_s])^2}{2}$. For the pair of edges that collide, 
		we can upper bound their contribution as
		\[
			\sum_{a\in A} \binom{\deg(a,B)}{2} \cdot  \frac{s^3}{|A|^2|B|} + \sum_{b\in B} \binom{\deg(b,A)}{2} \cdot  \frac{s^3}{|A||B|^2} < \frac{s^3}{2|A||B|}\left(\sum_{a\in A} \frac{\deg^2(a,B)}{|B|} + \sum_{b\in B} \frac{\deg^2(b,A)}{|A|}\right)
		\]
		which establishes \eqref{eq:var} after subtracting $(\Exp[Z])^2$.
		
		\noindent 
		Next, we use quasi-regularity and $\alpha$-density, to get 
		\[
		\sum_{a\in A} \frac{\deg(a,B)^2}{|B|} \le \sum_{a\in A} \frac{\deg(a,B)\deg(a)}{|B|} \leq \frac{2k\cdot m(A,B)}{|B|} 
		\]
		where the last inequaity used the upper bound $\deg(a) \leq 2k$ and $\sum_{a\in B} \deg(a,B) = m(A,B)$. 
		Using the lower bound, we get $k|A| \leq \sum_{a\in A} \deg(a) \leq 2m$ or $k\leq \frac{2m}{|A|} \leq \frac{2\alpha m(A,B)}{|A|}$, where
		we used the density of $(A,B)$ in the last inequality. 
		Plugging above, we get 
		$
		\sum_{a\in A} \frac{\deg(a,B)^2}{|B|} \le 4\alpha \frac{m^2(A,B)}{|A||B|}.
		$
		The same symmetric upper bound can be proved for $\sum_{b\in B} \frac{\deg(b,A)^2}{|A|}$.
		Substituting these in~\eqref{eq:var}, we get
		\[
			\Var[Z] \leq \Exp[Z] \cdot \left(1 - \frac{\Exp[Z]}{2}\right) + \frac{8\alpha}{s} \cdot \left(\frac{s^2}{|A||B|}\cdot  m(A,B)\right)^2 
		\]
		with~\Cref{clm:exp} completing the proof.
	\end{proof}
	
\noindent
If $s^2\mu \geq 2$, then $\Var[Z] \leq O(\alpha/s) (\Exp[Z])^2$ and therefore by Chebyshev $\Pr[Z = 0] \leq \frac{\Var[Z]}{(\Exp[Z]^2)} = O(\alpha/s)$.
This establishes part (c) of the lemma.

If $s^2\mu \geq 0.1$, we have $\Var[Z] \leq (10 + 8\alpha/s)(\Exp[Z])^2 < 18\alpha (\Exp[Z])^2$ (crudely). This is true for every $Z^{(i)}$ 
and so the variance of $Z_s$ is $\leq \frac{\eps^2}{10}\cdot (\Exp[Z_s])^2$. This shows that wp $\geq 9/10$, we have $Z_s \in (1\pm \eps) s^2\mu$.
 By~\Cref{lem:estrhos}, we that with probability $> 23/25$, $\widehat{\rho}_A \in (1\pm \eps)\rho_A$ and $\widehat{\rho}_B \in (1\pm \eps)\rho_B$.
 Union bounding over all errors, we get that with $> \frac{4}{5}$ probability we have
\[
	Z_s \in (1\pm \eps) \frac{s^2\rho_A\rho_B}{m(A,B)}, ~~\widehat{\rho}_A \in (1\pm \eps)\rho_A,~~\widehat{\rho}_B \in (1\pm \eps)\rho_B
\]
which implies $\widehat{m}(A,B) \in (1 \pm 3\eps)m(A,B)$.
\end{proof}

%
%

\subsection{The full algorithm for estimating average degree}\label{sec:3.2}

\def\guess{g}
\begin{algorithm}[ht!]
	\caption{{\sc \EstAvgDeg($\eps \in (0,1)$)}}\label{alg:est-d-known-n-full-alg}
	\begin{algorithmic}[1]
		\LineComment{Assume access to $\RV$, $\RN$, $\RE$, $\DEG$ and either $\ADD$ or weaker $\PAIR$.}
		\LineComment{$\eps \in (0,1)$ is a parameter; assumes knowledge of $n$.}
		\LineComment{Algorithm outputs estimate $\widehat{\davg}$ of average degree}
		\LineComment{Query Complexity: $O\left(\eps^{-2}\log^5 n\cdot \left(n/{\davg}\right)^{1/4}\right)$ with $\ADD$ and $O\left(\eps^{-2}\log^5 n\cdot \left(n/{\davg}\right)^{1/3}\right)$ with $\PAIR$.}	
		\Statex
		
		\State $\guess \eq n$ \Comment{$\guess$ is current guess of average degree}
		\State $p \eq (\guess/n)^{1/4}$ with $\ADD$ and $p\eq (\guess/n)^{1/3}$ with $\PAIR$. \label{alg:line:six}
		\While{True}:
			\State $(f, f', \Delta) \eq $ \hyperref[alg:esthighdeg]\AlgEstDensHighDeg($p$) \label{alg:line:esthighdeg} \Comment{Query Complexity: $O(\log n/p)$} 	
			\If{$f' < \frac{2}{3}$}: 
				\State \Return $\widehat{\davg} \eq $ \hyperref[alg:harmest]{\sc HarmonicEstimator}($\Delta, p, \eps$)  \label{alg:line:harm-est}\Comment{Query Complexity: $O(\frac{1}{\eps^2 p})$}
			\Else: \label{alg:line:else}
				\State $\alpha \eq \frac{\log^2 n}{16}$; use~\Cref{lem:recogij} to identify $i,j \in [\ceil{\log_2 n}]$ such that $H'_i, H'_j$ are $\alpha$-dense. \label{alg:line:twelve}
				\State Set $s \eq \ceil{\frac{1}{p}\sqrt{\alpha}}$ with $\ADD$; set $s \eq \ceil{\sqrt{\alpha/p}}$ with $\PAIR$
				\State output $\eq$ \hyperref[alg:est-mab]{\sc EstNumEdges}($H'_i, H'_j, 2^i(\Delta+1), 2^j(\Delta+1), s, \eps$) \LineComment{Query Complexity: $O(\frac{\alpha^2s}{\eps^2}) = O(\frac{\alpha^{2.5}}{\eps^2p})$ with $\ADD$; or
				$O(\frac{\alpha^2s^2}{\eps^2}) = O(\frac{\alpha^3}{\eps^2 p})$ with $\PAIR$}\label{alg:line:estnumedge}
				\If{output is $\abort$}:
					\State $\guess \eq \guess/2$; $p \eq (\guess/n)^{1/4}$; {\bf continue} to next while-loop
				\Else: \label{alg:siskin}
					\State Let $\widehat{m}$ be the output.
					\State Estimate $\widehat{\theta} \eq $ \hyperref[lem:ber-subs]{\EstFracEdge}($H'_i, H'_j$) \Comment{Query Complexity: $O(\alpha/\eps^2)$} 
					\State \Return $\widehat{\davg} \eq (\widehat{m}/\widehat{\theta}) \cdot \frac{2}{n}$
				\EndIf
			\EndIf
		\EndWhile
	\end{algorithmic}
\end{algorithm}

\noindent
The full algorithm is described in~\Cref{alg:est-d-known-n-full-alg}. 

\begin{theorem}\label{thm:est-d-known-n-n-by-d}
	\hyperref[alg:est-d-known-n-full-alg]{\sc \EstAvgDeg}($\eps$) returns an $(1\pm \eps)$-approximation to $\davg$ given access to $\RV$, $\RN$, $\RE$, $\DEG$ and either $\PAIR$ or $\ADD$ oracle.
	With $\PAIR$ oracle, it takes $\Ot((n/d)^{1/3})$-queries and with $\ADD$ oracle, it takes $\Ot((n/d)^{1/4})$ queries.
\end{theorem}

\noindent
Before proving the above theorem, note that along with~\Cref{lem:bt} this proves Part (a) of~\Cref{thm:stan} and~\Cref{thm:advanced}.

\begin{proof}[Proof of~\Cref{thm:est-d-known-n-n-by-d}]
{\em Time.} Let $g^\star$ be the random variable indicating the final value of $g$
and let $p^\star = (g^\star/n)^{1/4}$ (or $(g^\star/n)^{1/3}$ with $\PAIR$). The algorithm stops either at~\Cref{alg:line:harm-est} or because the output at~\Cref{alg:line:estnumedge} is not $\abort$, 
upon which it also runs the ``else-block'' from~\Cref{alg:siskin}. Since the variable $p$ drops geometrically, the total query complexity is dominated by the final iteration.
And so, the expected query complexity of the algorithm is $O(\eps^{-2}\alpha^{2.5}/p^{\star})$ (or $O(\eps^{-2}\alpha^{2.5}/p^{\star})$ with $\PAIR$). 
We now show with probability $1-o(1)$, the final $g^\star$ is $\Omega(\davg)$ and this proves the query complexity for $\ADD$ and $\PAIR$.
To see this, fix a $g$ and corresponding $p=(g/n)^{1/4}$,  and consider a run of~\Cref{alg:line:else}.
Let $\mu := \frac{m(H'_i:H'_j)}{|H'_i||H'_j|}$ and by~\Cref{clm:chern} we may assume that $|H'_i|+|H'_j| \leq 2np$ which implies $|H'_i||H'_j|\leq n^2p^2$.
Therefore, $\mu \geq \frac{m}{\alpha n^2p^2}$ since $(H'_i, H'_j)$ are $\alpha$-dense. Thus, $s^2\mu \geq \frac{m}{n^2p^4} = \frac{\davg}{2g}$ in case of $\ADD$ 
(and $s^2\mu \geq \frac{m}{n^2p^3} = \frac{\davg}{2g}$ in case of $\PAIR$).
By part (c) of~\Cref{lem:estimate-mab}, if $g < \davg/4$, the probability~\Cref{alg:line:estnumedge} returns abort is $O(\alpha^{0.5} p) = o(1)$ (assuming $\davg \leq n/\polylog n$).
And so, this iteration aborts with probability $1-o(1)$ proving $g^\star = \Omega(\davg)$, which proves the query complexity of our algorithm. \smallskip

\noindent{\em Correctness.}
First note that if we ever return using~\Cref{alg:line:harm-est}, then~\Cref{lem:harmest} implies the output $\widehat{\davg} \in (1\pm \eps)\davg$ with $>9/10$ probability. 
Suppose we returned because~\Cref{alg:siskin} didn't abort and let $g^\star$ be the value of $g$ at this point, and let $s,\mu$ as in the analysis of queries complexity in the previous paragraph.
If $s^2\mu < 0.1$, then by part (a) of~\Cref{lem:estimate-mab}, we would've aborted with $\geq 9/10$ probability. If $s^2\mu \geq 0.1$, then part (b) of the same lemma states that 
$\widehat{m}$ is a good estimate of $E(H'_i, H'_j) =: m_{ij}$ with probabilith $> $. 
Finally, $\widehat{\theta} \in (1 \pm \eps)\frac{m_{ij}}{m}$ with $1-o(1)$ probability. And this implies at the end, $\widehat{\davg}$ is in $(1\pm O(\eps))\davg$.
The total error probability can be union-bounded to $\leq \frac{1}{10} + \frac{1}{10} + \frac{1}{5} = \frac{2}{5}$.
\end{proof}

\subsection{When we can only sample edges but not vertices}\label{sec:3.3}
We end this section by noting that if only had the power to sample random edges but not random vertices, but still had the power of sampling random neighbors {\em and}
access to $\ADD$ queries, we can get $\Ot(n^{1/3})$-query algorithms assuming the graph has no isolated vertices. This is a restatement of~\Cref{thm:onlyedges}.

\begin{theorem}\label{thm:est-davg-known-n-no-rv}
	There is a $\Ot(n^{1/3})$ query randomized algorithm which returns an $(1\pm \eps)$-approximation to $\davg$ 
	of a simple undirected graph with no isolated vertices, when given access to $\RE$, $\RN$, $\DEG$ and $\ADD$ oracles.
	We assume the number of vertices is known.
\end{theorem}
\begin{proof}
	We first show a $\Ot(\eps^{-2}(n/d)^{1/2})$-query algorithm. To get this,
	run the \hyperref[alg:est-d-known-n-full-alg]{\EstAvgDeg}$(\eps)$ algorithm with the following changes: (i) in~\Cref{alg:line:six}, set $p\eq (\guess/n)^{1/2}$,  (ii) remove the if-then-else lines~\Cref{alg:line:esthighdeg} to~\Cref{alg:line:else}, and straightaway going to~\Cref{alg:line:twelve} (iii) use parameter $\Delta \eq 0$. Note that \Cref{lem:recogij}, \hyperref[lem:ber-subs]{\EstFracEdge}, and \hyperref[alg:est-mab]{\sc \EstNumEdges} 
	use only $\RE$, $\RN$, $\DEG$, and $\ADD$ oracles. In the analysis (proof of~\Cref{thm:est-d-known-n-n-by-d}), we use the (trivial) upper bounds on the set $|H'_i|, |H'_j|$ of $\leq n$
	giving $\mu \geq \frac{m}{\alpha n^2}$ and so $s^2\mu \geq \frac{m}{n^2p^2} = \frac{\davg}{2g}$. The rest of the analysis is the same. The query complexity is $O(\eps^{-2}\log^5 n \sqrt{n/\davg})$.
	
	To get the $\Ot(n^{1/3})$ query complexity, apply~\Cref{lem:bt}. More precisely, set $\theta = n^{1/3}$ and if~\Cref{lem:bt} succeeds, we are done. Otherwise, we know $\davg > n^{1/3}$
	and the algorithm described in the previous paragraph runs in $\Ot(\eps^{-2}\sqrt{n/d})$ queries.
\end{proof}

\section{Estimating average degree with unknown number of vertices} \label{sec:estd}

When the number of vertices is known, ERS give an algorithm {\sc \ERSEstAvgDeg} to $(1\pm \eps)$-approximate $\davg$
using only $\RV$, $\RN$ and $\DEG$ oracles (base model) in $O(\eps^{-2}\sqrt{n/d})$ queries~\cite{ERS19}. In particular, for $d = \omega(1)$, 
it is a $o(\sqrt{n})$-query algorithm. When $n$ is {\em not known}, one cannot obtain such a result; 
there is a lower bound of $\sqrt{n}$ even for high average degree graphs (\Cref{clm:apple} in~\Cref{sec:lowerbounds}).
In this section we show that the ERS bound can be recovered if we have access to either
(i) $\RE$ queries or (ii) $\ADD$ queries.

The main issue with running the ERS algorithm without any knowledge of $n$ or $\davg$ is that we don't know when to stop.
More precisely, ERS has an estimator that needs to be repeated $O(\sqrt{n/\davg})$ times. With no knowledge of $n$ (and of course none of $\davg$), there is no way to know this value $n/\davg$.
The ``guessing and checking'' methods do not work with two unknowns; neither does the idea of  stopping when empirical variance is low enough (an idea often used in practice) as it runs the risk of stopping too early. Indeed,~\Cref{clm:apple} shows there is no uniform way to fix this issue without assuming more powerful queries.

In~\Cref{sec:est-d-unknown-n-re}, we show how to bypass the need to know ``$n/\davg$'' using the ability to sample random edges, giving a $O(\eps^{-2}\sqrt{n/d})$ query algorithm 
to estimate the average degree. Combined with the $O(\eps^{-2}\davg)$ query algorithm {\sc \BTEstAvgDeg} (\Cref{lem:bt}) due to BT, we get a $O(\eps^{-2}n^{1/3})$ query algorithm.
In~\Cref{sec:est-d-unknown-n-add}, we show how to estimate either $n$ or $n/\davg$
in roughly $\sqrt{n/\davg}$ queries in the Advanced model ($\RE$ isn't required). This allows us to implement ERS.
An added benefit is that we can estimate the {\em number of vertices} in the same amount of queries.
In contrast, with only $\RV, \RN, \RE$ and $\DEG$, one can't beat ``birthday paradox bound'' for vertex estimation (\Cref{clm:cherry}).

\subsection{With access to uniformly random edges}\label{sec:est-d-unknown-n-re}\label{sec:4.1}

The algorithm 
described below obtains the requisite ``lower bound'' on the number of samples by designing a geometric random variable with a ``small enough'' success probability.
More precisely, the algorithm below computes a parameter $\tau$, and 
repeatedly runs the ERS estimator  
till the estimator crosses $\tau$ enough ($\approx 1/\eps^2$) times.
We use $\RE$ to choose the $\tau$: we sample $O(1/\eps)$ random edges and choose $\tau$ using the degrees of the vertices seen. 
We now give the details and the analysis.

\renewcommand{\dest}{{\sc EstAvgDeg}}
\begin{algorithm}
	\caption{\dest($\eps \in (0,1)$)} \label{alg:dest}
	\begin{algorithmic}[1]
		\LineComment{Assume access to $\RV$, $\RN$, $\DEG$, and $\RE$.}
		\LineComment{ $\eps \in (0,1)$ is a parameter. No other assumptions made}
		\LineComment{Returns $\widehat{\davg}$}
		\LineComment{Query Complexity: $O(\eps^{-2.5}\sqrt{n/\davg})$ but see remark after \Cref{lem:dest-conc}.}
		\Statex 
		
		\State  Sample $10/\eps$ uar edges $e\in E$ and call this set $R$. \Comment{Uses $\RE$}
		\State $\tau \eq \max_{e\in R} \deg(e)$ \Comment{Recall: for $e=(x,y)$, $\deg(e) = \min(\deg(x), \deg(y))$} \label{alg:settau} \Comment{Uses $\DEG$}
		\State Initialize $Y \eq 0$; $k\eq 0$
		\While{$Y < 200/\eps^2$}:
			\State $k\eq k+1$
			\State Sample $x \in V$ uar. \Comment{Uses $\RV$}
			\If{$\deg(x) \geq \tau$}:\label{alg:no-adv:estd:degcheck}
				\State $Y \eq Y + 1$
			\EndIf
			\State Sample $y\sim x$ uar. \Comment{Uses $\RN$}
			\If{($\deg(y) > \deg(x)$) or ($\deg(y) = \deg(x)$ and $\id(y) > \id(x)$)}: \Comment{Uses $\DEG$}	
				\State $X_k \eq \deg(x)$
			\Else:
				\State $X_k \eq 0$
			\EndIf
		\EndWhile
		\State \Return $\widehat{\davg} \eq  \frac{2}{k}\sum_{i=1}^k \min(X_i, \tau)$
	\end{algorithmic}
\end{algorithm}
\begin{theorem}\label{thm:estavgdeg-re-rn-pair}
	\hyperref[alg:est-d-known-n-full-alg]{\dest$(\eps)$} returns an $(1\pm \eps)$-approximation to $\davg$ given access to $\RV$, $\RN$, $\RE$ and $\DEG$ oracles, in $O(\eps^{-2.5}\sqrt{n/d})$ queries.
	Using this, one can obtain a $O(\eps^{-2}\sqrt{n/d})$ query algorithm as well.
\end{theorem}
\begin{proof}
Given $\tau$, let us define two subsets of vertices:
\[
H := \{x\in V~:~\deg(x) \geq \tau\}~~\text{and}~~H' := \{x\in V~:~\deg(x) \geq \tau+1\}
\]
\begin{claim}\label{clm:obs-estd}
	With probability $\geq 9/10$, we have  $|E[H']| \leq \eps m/4$.
	Furthermore, the probability $|H| < \sqrt{\delta\eps m}$ is $\leq 10\delta$, for any $\delta$.
\end{claim}
\begin{proof}
Consider the edges in decreasing order of their degrees (breaking ties arbitrarily). The probability that the random sample $R$
{\em does not} contain any of the first $\eps m/4$ edges is $(1 - \eps/4)^{10/\eps} < e^{-5/2} < 1/10$. So, with $\geq 9/10$ probability, $\tau \geq \deg(e)$ for 
at least $(1-\eps/4)m$ edges. These edges must have
at least one endpoint must have $\deg(x) \leq \tau$, that is, this endpoint lies outside $H'$. So, the number of edges with both endpoints in $H'$, that is $|E[H']|$, is $\leq \eps m/4$.

Similarly, the probability the set $R$ does not contain any of the first $\delta \eps m$ edges is $(1 - \delta \eps)^{10/\eps} > 1 - 10\delta$.
That is, with this probability, at least $\delta \eps m$ edges $e$ have $\deg(e) \geq \tau$. The endpoints of these edges are all in $H$ (since they all have degree $\geq \tau$),
and so there are at least $\sqrt{\delta \eps m}$ vertices in $H$. So, $\Pr[|H| < \sqrt{\delta \eps m}] \leq 10\delta$.
\end{proof}

\begin{lemma}
	 \label{clm:run-approx-d} 
	 With probability $\geq 9/10$, the final value of $k$ is at least $\frac{50\tau}{\eps^2 \davg}$.
	 Furthermore, 
	the expected query complexity of \dest{} is $O(\eps^{-2.5} \sqrt{n/d})$. 

\end{lemma}
\begin{proof}
	Consider the increment of $Y$ in~\Cref{alg:no-adv:estd:degcheck}. This is a geometric random variable with parameter $|H|/n$ and the query complexity/number of samples 
	is the number of trials till a sum of geometric random variables exceeds $200/\eps^2$.
	First we observe that $|H|/n \leq \davg/\tau$ since $|H|\tau \leq \sum_{x\in H}\deg(x) \leq 2m$. This implies the lower bound on the number of samples by a Chernoff bound calculation\footnote{The probability of seeing $200/\eps^2$ heads in less than $\frac{50\tau}{\eps^2\davg}$ tosses of a coin with bias $\leq \davg/\tau$ is $< 2^{-4}$.}.
	
	Using~\Cref{clm:obs-estd}, we can establish the expected query complexity. Note if $|H| \in (\sqrt{2^{-(i+1)}\eps m}, \sqrt{2^{-i}\eps m}]$, we can upper-bound the expected query complexity/final value of $k$ by $\frac{200}{\eps^2}\cdot \frac{n}{\sqrt{2^{-(i+1)}\eps m}}$.
	Therefore, the expected value of $k$ can be upper bounded by
	\[
		\Exp[k] \leq \sum_{i=0}^\infty \Pr[|H| \in (\sqrt{2^{-(i+1)}\eps m}, \sqrt{2^{-i}\eps m} ] \cdot \frac{200}{\eps^2}\cdot \frac{n}{\sqrt{2^{-(i+1)}\eps m}} \leq \frac{200}{\eps^{2.5}}\sqrt{n/\davg} \sum_{i=0}^\infty \frac{10\cdot 2^{-i}}{2^{-(i+1)/2}} 
	\]
	which evaluates to  $O(\eps^{-2.5} \sqrt{n/d})$
\end{proof}

\begin{lemma} \label{lem:dest-conc} With probability at least $2/3$, the output $\widehat{\davg}$
	of \dest{} is in $(1\pm \eps) \davg$.
\end{lemma}

\begin{proof} 
	We condition on the events $|E[H']| \leq \eps m/4$ and $k \geq \frac{50\tau}{\eps^2 d}$; by~\Cref{clm:obs-estd} and~\Cref{clm:run-approx-d}, these occur with probability $4/5$.
	
	\noindent
	Consider an orientation of the edges of $G$ such that given $(x,y)$ we point it from 
	$x$ to $y$ if $\deg(y) > \deg(x)$ or $\deg(x) = \deg(y)$ and $\id(y) > \id(x)$.
	Let $\deg^+(x)$ denote the out-degree of $x$ in this orientation. 
	Note that
		$\sum_{x\notin H'} \deg^+(x) \geq (1-\eps/4)m$
	since the only edges not counted in the LHS are the ones in $E[H']$.
	For convenience, let $W_i := \min(X_i, \tau)$. 
	\[
		\Exp[W_i] \geq \sum_{v \notin H'} \Pr[x = v]\cdot \frac{\deg^+(x)}{\deg(x)} \cdot \min(\deg(x), \tau) = \frac{1}{n} \sum_{x\notin H'} \deg^+(x) \geq (1 - \eps/4)\frac{m}{n}
	\]
We also have $\Exp[W_i] \leq \sum_{v \notin H'} \Pr[x = v]\cdot \frac{\deg^+(x)}{\deg(x)} \cdot \deg(x) = \frac{m}{n}$. Thus, $2\Exp[W_i]$, although not unbiased, is a pretty good estimate of $\davg = 2m/n$.
Let's use $\mu$ to denote $\Exp[W_i]$ (which is the same for all $i$), and so $2\mu \in \left((1-\eps/4)\davg, \davg\right]$.
Next, using the fact that the $W_i$'s are bounded by $\tau$, we get that $\Var[W_i] \leq \tau\mu$. Therefore, if $\sigma^2$ is the variance of $\widehat{\davg}$, we get that $\sigma^2 \leq \frac{4\tau\mu}{k}$.
So, Chebyshev gives 
\[
	\Pr[\left|\widehat{\davg} - 2\mu\right| \geq \eps \mu] \leq \frac{\sigma^2}{\eps^2 \mu^2} \leq \frac{4\tau}{k\eps^2 \mu} < \frac{1}{20}
\]
Since $2\mu \in [(1-\eps/4)\davg, \davg]$, we get that with probability all but $\frac{1}{20} + \frac{1}{10} + \frac{1}{10} = \frac{1}{4}$, our estimate $\widehat{\davg} \in (1\pm \eps)\davg$.\qedhere
\end{proof}
\noindent

\paragraph{Improving the dependence on $\eps$.}\label{par:owl}

The proof of~\Cref{lem:dest-conc} shows that if we took $\frac{50\tau}{\eps^2 \davg}$ samples of $\min(X_i,\tau)$ where $\tau$ is defined as in~\Cref{alg:settau}, then $\widehat{\davg}$ 
is a $(1\pm \eps)$-estimate of $\davg$ with $> 3/4$ probability. Since we didn't know $\davg$, we used the geometric variable of~\Cref{alg:no-adv:estd:degcheck} to ensure we took enough samples. Unfortunately, this leads to an extra $1/\sqrt{\eps}$ in the expected query complexity
which is undesirable. To fix this, one runs the above algorithm with a constant $\eps_0 = 0.5$, say. This gives an estimate of $\widehat{\davg}_0 \in \left[\frac{\davg}{2}, 2\davg\right]$ in $O(\sqrt{n/\davg})$ queries. Once we have that, 
we can replace the while-loop with a for-loop of $\frac{100\tau}{\eps^2 \widehat{\davg}_0}$ iterations and delete the if-condition of~\Cref{alg:no-adv:estd:degcheck}. 
This concludes the proof of~\Cref{thm:estavgdeg-re-rn-pair}.\qedhere\end{proof}

\noindent
\Cref{lem:bt} and~\Cref{thm:est-davg-known-n-no-rv} proves~\Cref{thm:unknown-base}.

\subsection{With access to additive queries}
\label{sec:est-n} \label{sec:est-d-unknown-n-add}

\noindent
Consider the Advanced model (but not $\RE$).
We describe an algorithm {\sc \EstSpar} which, in $O(\eps^{-2}\sqrt{n/\davg})$ queries, either returns an estimate $\widehat{n}$ of the number of vertices, or return $\widehat{\rho}$ of  the inverse
density $\rho = n/\davg$ of a graph. In either case, we can run \hyperref[alg:ers]{\sc \ERSEstAvgDeg} algorithm of~\cite{ERS19} (see~\Cref{rem:ers})  to get an estimate $\widehat{\davg}$ of the average degree in $O(\eps^{-2}\sqrt{n/\davg})$. Also note we can estimate the number of vertices; either we get it directly, or we can obtain by  $\widehat{n} = \widehat{\rho}\widehat{\davg}$.
This proves part (a) of~\Cref{thm:unknown-no-edge}. In~\Cref{sec:lowerbounds}, we show these bounds are optimal. \smallskip

The idea behind the algorithm is simple. We sample a subset $X$ of $s$ nodes where $s$ is a parameter which grows exponentially as $\{2,4,8,\cdots\}$, and use $\ADD$ to obtain $m(X)$, the number of edges
with both endpoints in $X$, in $O(s)$ queries. It is not hard to see that the expected value of $m(X)$ is roughly $s^2\davg/2n$, and so if $s \approx \sqrt{n/\davg}$ {\em and} the variance of $m(X)$ is low, 
then we would get a good estimate to $n/\davg$. The variance is controlled by $D^{(2)} := \sum_{x\in V} \deg^2(x)$, and if this is $O(m^{1.5})$), then convergence occurs.
Unfortunately, this may not be the case. We then use a combinatorial lemma (\Cref{lem:onehop}) that says that if the sum-of-squares of degrees is ``high'', then a random vertex 
or a random neighbor of a random vertex will have ``high'' degree ($> \sqrt{m})$ with probability $\geq \sqrt{d/n}$. Suppose our $m(X)$ estimator is bad. Then sampling $O(\sqrt{n/d})$ vertices
and their random neighbors, and picking the vertex among them with the largest degree, would whp have degree $\geq \sqrt{m}$.
With such a vertex $x$, we can then estimate the number of vertices by rejection sampling (\Cref{lem:ber-subs}). Thus, we either get a good estimate of $\rho$ or a good estimate of $n$.
In the pseudocode below,~\Cref{alg:line:14,alg:line:23} implement the first idea, while~\Cref{alg:line:6,alg:line:12} implements the second.

\begin{algorithm}[ht!]
	\caption{{\sc \EstSpar($\eps \in (0,1)$)}}\label{alg:est-spar}
	\begin{algorithmic}[1]
		\LineComment{Assume access to $\RV, \RN$, $\DEG$, and $\PAIR$/$\ADD$. }
		\LineComment{$\eps \in (0,1)$ is a parameter. No other assumptions made.}
		\LineComment{Returns: either an estimate $\widehat{n}$ of $n$, or estimate  $\widehat{\rho}$ of $n/\davg$}
		\LineComment{Query Complexity: with $O(\eps^{-2} (n/\davg))$ with $\PAIR$, and $O(\eps^{-2} \sqrt{n/\davg})$ with $\ADD$}
		\Statex
		
		\State $s\eq 2$
		\State Sample $A\subseteq V$ of $s$ uar nodes from $V$ \Comment{Takes $O(s)$ queries using $\RV$} \label{alg:est-n:l1}\label{alg:line:6}
		\For{$a\in A$}: \label{alg:est-n:line:deg-idea-begin}
			\State Sample $b \sim a$ and add it to a set $B$. \Comment{Uses $\RN$}
		\EndFor
		\State $x \eq \arg \max_{v\in A\cup B} \deg(v)$ \Comment{Uses $\DEG$} 
		\State Sample $X\subseteq V$ of $s$ uar nodes. \Comment{Uses $\RV$} \label{kookaburra}
		\If{$x$ has a neighbor in $X$}: \Comment{Takes $O(s)$ queries either using $\PAIR$}\label{alg:est-n:line:nbrchck} \label{duck}
			\State \Return $\widehat{n} \eq $ \hyperref[lem:ber-subs]{\EstNumVer}$(x)$  \label{alg:est-n:l:find-v} \label{alg:est-n:line:deg-idea-end} \label{mallard} \Comment{Return estimate of $n$}\label{alg:line:12}
		\Else:
			\State Query $m(X)$ \Comment{Takes $O(s^2)$ queries using $\PAIR$ and $O(s)$ queries using $\ADD$}  \label{alg:line:14}
			\If{$m(X) = 0$}: \label{alg:est-n:checkzero}
				\State $s\eq 2s$ and go to~\Cref{alg:est-n:l1}
			\Else: \Comment{Assert low variance of estimator (see analysis)} \label{alg:est-n:l:else}
				\State $k \eq \ceil{\frac{A}{\eps^2}}$ \Comment{Think of $A$ is a rather large constant (see~\Cref{lem:low-var} for details)}
				\For{$i \in \{1,2,3,\ldots, k\}$}:
					\State Sample $X_i$ uar of size $|X_i| = s$. 
					\State Query $m^{(i)}_s \eq m(X_i)$ 
				\EndFor
				\State $\widehat{m}_s \eq \frac{1}{k}\sum_{i=1}^k m^{(i)}_s$
				\State \Return $\widehat{\rho} \eq s^2/2\widehat{m}_s$ \Comment{Return estimate of $\rho = (n/\davg)$.} \label{alg:line:23}
			\EndIf
		\EndIf
	\end{algorithmic}
\end{algorithm}
\noindent
\begin{theorem}\label{thm:est-dens-unknown-n}
	\hyperref[alg:est-spar]{\sc \EstSpar}$(\eps)$ either returns an $(1\pm \eps)$-approximation $\widehat{n}$ to $n$, or returns an $(1\pm \eps)$-approximation $\widehat{\rho}$ to $n/\davg$.
	It uses $\RV$, $\RN$, $\DEG$ and either $\PAIR$ or $\ADD$ oracle.
	With $\PAIR$ oracle, it takes $O(\eps^{-2}(n/\davg))$ queries and with $\ADD$ oracle, it takes $O(\eps^{-2}\sqrt{n/\davg})$ queries.
\end{theorem}
\noindent
Before proving the above, lets state a corollary for the standard model with unknown $n$.
\begin{corollary}\label{cor:unknown-n-no-edge}
	Consider the standard model with $n$ unknown. 
	There is an algorithm that outputs a $(1\pm \eps)$-approximation to $\davg$ in $\Ot\left(\eps^{-2}\min(\sqrt{n},n/\davg)\right)$ queries.
\end{corollary}
\begin{proof}
	In $O(\eps^{-2}\sqrt{n})$ queries, we can obtain an $(1\pm \eps)$-approximation of $n$ and run \hyperref[alg:ers]{\sc \ERSEstAvgDeg} (see~\Cref{rem:ers}).
	Or, we can run \hyperref[alg:est-spar]{\sc \EstSpar}$(\eps)$ to obtain either an approximation of $n$ or $n/d$ in $O(\eps^{-2}(n/d))$ queries, and again run \hyperref[alg:ers]{\sc \ERSEstAvgDeg} (see~\Cref{rem:ers}).
\end{proof}
\begin{proof}[Proof of~\Cref{thm:est-dens-unknown-n}]
We begin with analysis of~\Cref{kookaburra,duck,mallard}. 
\begin{claim}\label{clm:smalls}
	If the algorithm returns an estimate $\widehat{n}$ in~\Cref{alg:est-n:line:deg-idea-end}, it is an $(1\pm \eps)$-estimate of $n$. The expected queries made
	is $O(s^\star/\eps^{2})$ where $s^\star$ is the current (and final) value of the variable $s$.
\end{claim}
\begin{proof}
	The accuracy of the estimate follows from~\Cref{lem:ber-subs}.
    Suppose~\Cref{alg:est-n:line:nbrchck} returned True. The probability of this event 
    is $1 - \left(1 - \frac{\deg(x)}{n}\right)^s \leq \frac{s\deg(x)}{n}$.
	The estimated query complexity for \EstNumVer$(x)$  is $O(\frac{n}{\eps^2\deg(x)})$. Therefore conditioned on~\Cref{alg:est-n:line:nbrchck} returning true
	the expected query complexity is $O(\frac{n}{\eps^2\deg(x)}\cdot \frac{s\deg(x)}{n}) = O(s/\eps^2)$.
\end{proof}

\noindent
Next, we understand the mean and variance of the  random variable $Z = m(X)$ for a particular $s\in \NN$. 
Note that each $m^{(i)}$ is an independent draw from this distribution. 
These calculations below are similar to~\Cref{clm:exp} and~\Cref{clm:var-calc} presented slightly differently. 
\begin{claim}\label{clm:exp-est-n}
	$\Exp[Z] = \binom{s}{2} \cdot \frac{\davg}{n}$
\end{claim}
\begin{proof}
	Let $X_{ij}$ be the indicator variable that the $i$th and $j$th sampled vertex form an edge. Note that $Z = \sum_{1\leq i < j \leq s} X_{ij}$.
	Furthermore, 
	\[\Pr[X_{ij} = 1] = \sum_{x\in V} \Pr[i\text{th sample}~=x]\Pr[j\text{th sample}\sim x] = \sum_{x\in V} \frac{1}{n}\cdot \frac{\deg(x)}{n} = \frac{2m}{n^2} = \frac{\davg}{n}\]
	The claim follows by linearity of expectation.
\end{proof}

\begin{claim}\label{clm:var-calc-est-n}
	$\Var[Z] \leq \Exp[Z] + 2 \binom{s}{3} \cdot \frac{1}{n^3} \cdot D^{(2)}$ where $D^{(2)} := \sum_{x\in V} \deg^2(x)$.
\end{claim}
\begin{proof}
	As is usual, we bound $\Exp[Z^2]$ by computing $\Exp[\left(\sum_{i,j} X_{ij}\right)^2]$.
	We see that $\Exp[X_{ij}X_{k\ell}]$ for distinct $i,j,k,\ell$, by independence, equals $\left(\frac{n}{\davg}\right)^2$, 
	and summing over all\footnote{the more correct bound would be $3\binom{s}{4}$} $< \left(\binom{s}{2}\right)^2$ such pairs 
	would give $< \left(\Exp[Z]\right)^2$. What remains are the terms of the form $\Exp[X_{ij}X_{ik}]$ (and $\Exp[X_{ik}X_{jk}]$ but their calculation is symmetric leading to the factor $2$).
	We can evaluate this similar to the expectation calculation as
	\[
		\Exp[X_{ij}X_{ik}] = \sum_{x\in V} \Pr[i\text{th sample}~=x] \cdot \Pr[j\text{th and~}k\text{th samples}\sim x] = \sum_{x\in V} \frac{1}{n}\cdot \left(\frac{\deg(x)}{n}\right)^2 = \frac{D^{(2)}}{n^3}
	\]
	The claim now follows by linearity of expectation going over the $\binom{s}{3}$ triples and the two mappings ($ij,ik$ and $ik,jk$).
\end{proof}
\noindent
The variance has the sum-of-square of degrees bound, which is large if the degree distribution is skewed. 
In that case, we would like to argue that~\Cref{kookaburra,duck,mallard} will give the correct estimate for ``small'' $s$ itself.
The following combinatorial lemma, whose proof we defer, is key to this.

\begin{mdframed}[backgroundcolor=yellow!10,topline=false,bottomline=false,leftline=false,rightline=false] 
\begin{lemma} \label{lem:onehop}
	Let $G = (V,E)$ be an simple undirected graph on $n$ vertices and $m$ edges. Let $C$ be a parameter which should be thought of as a constant.
Let $H := \{x\in V:\deg(x) \geq C \sqrt{m}\}$. \smallskip

\noindent
If $D^{(2)} := \sum_{x\in V} \deg^2(x) \geq  (5C + 2C^3)m^{3/2}$  and $|H| < C\sqrt{m}$, ~~then $\sum_{v\in V} \frac{\deg(v,H)}{\deg(v)} > C\sqrt{m}$.	
\end{lemma}
\end{mdframed}
\noindent
Let us now complete the analysis of the algorithm. We begin with a claim that for ``small'' $s$,
\Cref{alg:est-n:checkzero} will return true whp. 
Henceforth, $\delta$ is a small constant; our probability of failure would be $\leq O(\delta)$.
This dictates the parameter $A$ which would roughly come out to be  $\frac{1}{\delta^3}$ (see \Cref{lem:low-var}).
\begin{claim}\label{cor:exp}
	Let $s^\star$ be the final value of $s$. If $s^\star \leq \sqrt{\delta n/d}$, then with $1-\delta$ probability, the {\bf return} statement is~\Cref{mallard}.
\end{claim}
\begin{proof}
	By Markov's inequality (\Cref{fact:mkov-cheb}), $\Pr[m(X) \neq 0] \leq \Exp[m(X)] = \frac{s^2\davg}{2n}$. Union bounding this for $s\in \{1,2,4,\ldots, s^\star\}$, 
	gives a geometric sum showing that the probability ~\Cref{alg:est-n:checkzero} returns true in any of these iterations is $\leq 2(s^\star)^2\davg/2n < \delta$.
\end{proof}
\noindent
\def\Cd{C_\delta}
\def\Bd{B_\delta}
The next claim uses~\Cref{lem:onehop} to show that if $D^{(2)}$ is large, then $s^\star$ is indeed ``small''. Along with the previous claim and~\Cref{clm:smalls}, we would take care of this ``large variance'' case.
With hindsight, define $C = \Cd := \ceil{\frac{3}{\sqrt{\delta}}\ln(1/\delta)}$ and $\Bd := 5C + 2C^3$.
\begin{claim}\label{clm:anal-onehop}
	If $D^{(2)} \geq \Bd \cdot m^{3/2}$, then with $1-2\delta$ probability $s^\star \leq \sqrt{\delta n/\davg}$.
\end{claim}
\begin{proof}
	Fix $s$ to be the largest power of $2$ smaller than $\sqrt{\delta n/\davg} = \frac{n \sqrt{\delta}}{\sqrt{2m}}$. So, $s\geq \frac{n\sqrt{\delta}}{3\sqrt{m}}$.
	We argue that \Cref{duck} returns True with probability $1-o(1)$ for such an $s$. This will prove the claim. To that end, consider
	the set $H$ as defined in~\Cref{lem:onehop} with the parameter $C$ defined before the claim.\smallskip
	
	\noindent
	Case 1: $|H| \geq C\sqrt{m}$. In this case, consider the set $A$ sampled in~\Cref{alg:est-n:line:deg-idea-begin}; note that
	$\Pr[A\cap H = \emptyset] = \left(1 - \frac{|H|}{n}\right)^s \leq e^{-s|H|/n} \leq  e^{-Cs\sqrt{m}/n} \leq \delta$ since $s\geq \frac{n\sqrt{\delta}}{3\sqrt{m}}$ (this explains the choice of $C$).
	Thus the $x$ in~\Cref{kookaburra} has $\deg(x) \geq C\sqrt{m}$. The probability that~\Cref{duck} returns False is $\leq (1 - \deg(x)/n)^s \leq e^{-s\deg(x)/n} \leq e^{-Cs\sqrt{m}/n} \leq \delta$.
	So, in this case we are done by union-bounding the two failure probabilities. \smallskip
	
	\noindent
	Case 2: $|H| < C\sqrt{m}$. 
	By~\Cref{lem:onehop}, since $D^{(2)}$ is ``large'', we have that
		$\sum_{a\in V} \frac{\deg(a,H)}{\deg(a)} \geq C\sqrt{m}$.
	Next let us examine the probability $\Pr[B\cap H = \emptyset]$. Note that $B$ can be thought of as $s$ independent experiments where we draw a uar vertex $v$ and then draw a uar neighbor $u$.
	The probability such a $u\in H$ is
	\[
		\Pr[u\in H] = \sum_{v\in V} \Pr[v~\text{sampled}]\cdot \frac{\deg(v,H)}{\deg(v)} \geq \frac{C\sqrt{m}}{n}
	\]
	Thus, $\Pr[B\cap H = \emptyset] \leq \left(1 - \frac{C\sqrt{m}}{n}\right)^s \leq e^{-Cs\sqrt{m}/n} \leq \delta$. Thus, as before, the $x$ in~\Cref{kookaburra} has $\deg(x) \geq C\sqrt{m}$, 
	and this shows the proof in this case as well.
\end{proof}
\begin{claim}\label{lem:high-var}
	When $D^{(2)} \geq B_\delta \cdot  m^{1.5}$, with $\geq 1-3\delta $ probability the algorithm returns an $(1\pm \eps)$-approximation to $n$ in $O(\eps^{-2}\sqrt{n/d})$ expected queries.
\end{claim}
\begin{proof}
	Follows from~\Cref{clm:anal-onehop},~\Cref{cor:exp}, and~\Cref{clm:smalls}.
\end{proof}

\begin{claim}\label{lem:low-var}
	When $D^{(2)} \leq \Bd \cdot m^{1.5}$, with $1-o(1)$ probability the algorithm either returns an $(1\pm \eps)$-approximation to $n$ 
	or an $(1\pm \eps)$-approximation to $n/\davg$ in $O(\eps^{-2}\sqrt{n/d})$ expected queries.
\end{claim}
\begin{proof}
	Let $s^\star$ be the final $s$.
	By~\Cref{cor:exp} and~\Cref{clm:smalls}, we may assume  $s^\star > \sqrt{n\delta /\davg}$.
When $s$ is this ``large'', we can argue that $\Var[Z] \approx (\Exp[Z])^2$. More precisely, from~\Cref{clm:exp-est-n} and~\Cref{clm:var-calc-est-n}, we get
\[ 
\frac{\Var[Z]}{(\Exp[Z])^2} \leq \frac{1}{\Exp[Z]} + \frac{s^3 D^{(2)}}{3n^3} \cdot \frac{4n^2}{s^4 \davg^2} ~~\leq~~ \frac{2n}{s^2 d} + \frac{4\Bd}{3s} \cdot \frac{m^{1.5}}{n\davg^2} ~< \frac{2}{\delta} + \frac{B_{\delta}}{\sqrt{\delta}}
\]
where the first inequality uses $D^{(2)}$ is ``small'', the second inequality uses $m = n\davg/2$ and that $s > \sqrt{\delta n/d}$.
Note $B_\delta \approx \delta^{-1.5}\ln^3(1/\delta)$; and so the ratio of variance by squared expectation is $\Ot(\frac{1}{\delta^2})$.
Since we are repeating $O\left(\frac{A}{\eps^2}\right)$ times iid, if $A \approx \frac{1}{\delta^3 \ln^3(1/\delta)}$, then averaging would drop the ratio to $\eps^2/\delta$ and then Chebyshev gives us that our estimate is $(1\pm \eps) \frac{s^2 d}{n}$
with all but $\delta$ probability.
\end{proof}
\noindent
\Cref{lem:high-var} and~\Cref{lem:low-var} proves \Cref{thm:est-dens-unknown-n}.
\end{proof}
\noindent
We end this section with the proof of the combinatorial lemma.

\begin{proof}[\bf Proof of~\Cref{lem:onehop}]
	Let $S = V\setminus H$ and by definition $\deg(x) < C\sqrt{m}$ for all $x\in S$. First. observe that
	\begin{equation}\label{eq:001}
	\sum_{x\in S} \deg^2(x) < C\sqrt{m} \sum_{x\in S} \deg(x) \leq Cm^{1.5} ~~\Rightarrow~~ \sum_{x\in H} \deg^2(x) > D^{(2)} - Cm^{1.5}
	\end{equation}
	\noindent
	Next, observe that
	\[
	\sum_{x\in H} \deg^2(x) = \sum_{x\in H} \left(\deg(x,S) + \deg(x,H)\right)^2 \leq 2\left(\sum_{x\in H} \deg^2(x,S) + \sum_{x\in H}\deg^2(x,H)\right)
	\]
	Since the graph is simple, we have $\deg(x,H) < |H|$. Plugging above gives us
	\[
		\sum_{x\in H}\deg^2(x,S) \geq \frac{1}{2} \left(\sum_{x\in H} \deg^2(x) - 2|H|^3\right) > \frac{1}{2} \cdot \left(D^{(2)} - (C + 2C^3)m^{1.5}\right)
	\]
	where the last inequality follows from\eqref{eq:001} and since $|H| < C\sqrt{m}$. Using the bound on $D^{(2)}$, we get
	
	\begin{equation}\label{eq:premise}
		\sum_{x\in H} \deg^2(x,S) > 2Cm^{1.5}
	\end{equation}
	\noindent
	Next note that
	\[
		\sum_{x\in H} \deg^2(x,S) = \sum_{x\in H} \left(\sum_{v\in S} \bone(x,a)\right)^2 = \sum_{x\in H} \sum_{v\in S} \bone(x,v)^2 + \sum_{x\in H} \sum_{u,v\in S, u \neq v} \bone(x,u)\bone(x,v)
	\]
	where $\bone(x,v)$ is the indicator whether $(x,v)\in E$ and so $\bone^2 = \bone$. Rearranging, we get
	\def\comm{\mathsf{comm}}
	\[
	\sum_{x\in H} \deg^2(x,S) = \sum_{v\in S} \deg(v,H) + \sum_{u,v\in S, u\neq v} \comm(u,v;H)
	\]
	where $\comm(u,v;H)$ is counting the number of common neighbors of $u$ and $v$ in $H$. 
	Note that $\comm(u,v;H) \leq \min(\deg(u,H), \deg(v,H)) \leq \sqrt{\deg(u,H)\deg(v,H)}$. Plugging above we get
	\[
	\sum_{x\in H} \deg^2(x,S) \leq \sum_{v\in S} \deg(v,H) + \sum_{u,v\in S;u\neq v} \sqrt{\deg(u,H)\deg(v,H)} = \left(\sum_{v\in S} \sqrt{\deg(v,H)}\right)^2
	\]
	Using  Cauchy-Schwarz we get
	\[
			\sum_{x\in H} \deg^2(x,S)  \leq \left(\sum_{v\in S} \sqrt{\deg(v,H)}\right)^2 = \left(\sum_{v\in S} \sqrt{\frac{\deg(v,H)}{\deg(v)}}\cdot \sqrt{\deg(v)} \right)^2\leq \left(\sum_{v\in S} \frac{\deg(v,H)}{\deg(v)}\right)\left(\sum_{v\in S} \deg(v)\right)
	\]
	The proof follows from~\eqref{eq:premise} and the fact that $\sum_{v\in S} \deg(v) \leq 2m$.
\end{proof}

\section{Lower Bounds}\label{sec:lowerbounds}
In this section, we prove matching lower bounds for our algorithms. We begin with the case when the algorithm knows the value of $n$, the number of vertices.

\begin{theorem}\label{thm:lb-known-n}
	Suppose there is a randomized algorithm which has access to $\RV$, $\RN$, $\DEG$, and $\RE$ oracle on an unknown graph $G=(V,E)$ {\em and} knows the number $n$ of vertices
	and with $>2/3$ probability returns a $O(1)$-approximation to $\davg$. 
	\begin{enumerate}[label=(\alph*),noitemsep]
		\item If there is no other oracle access, then the algorithm must make
        $\Omega(\min(d, \sqrt{n/d}))$ queries, which is  $\Omega(n^{1/3})$ for worst-case $d$.
		\item If it has access to $\PAIR$ oracle, then the algorithm must make 
        $\Omega(\min(d, \sqrt[3]{n/d}))$ queries, which is  $\Omega(n^{1/4})$ for worst-case $d$.
		\item If it has access to $\FN$ oracle, then the algorithm must make
        $\Omega(\min(d, \sqrt[4]{n/d}))$ queries, which is  $\Omega(n^{1/5})$ for worst-case $d$.
	\end{enumerate}
\end{theorem}
\begin{proof}
The structure of the bad example is the same for all the bounds differing only in parameters. Suppose the algorithm is a $c$-approximation. 
The graph $G = (V,E)$ is partitioned into $G_A \cup G_B$ where $G_A$ is a collection of $\gamma k$ cliques of size $s$, and $G_B$ is a matching of $n - \gamma sk$ vertices.
The parameter $\gamma \in \{1, c+1\}$ and the two choices lead to two different cases: the YES and NO, respectively. The number of edges in the graph is $m  = \gamma k\cdot  \binom{s}{2} + \frac{n-\gamma sk}{2}$.
In all cases, the setting of $s,k$ will be such that $sk \ll n$ and $s^2k \gg n$. Thus, $\davg = \frac{\gamma s^2 k}{2n}\cdot (1+o(1))$. Thus, any $c$-approximation must be able to distinguish between the two 
graphs. We now show if it makes fewer than the asserted queries, the probability of distinguishing is $o(1)$. \smallskip

\noindent
Below, let $Q_1$ be the vertices from $\RV$ and $Q_2$ be the endpoints of $\RE$. Two observations would explain the parameters of $s, k$. Suppose the algorithm makes $q$ queries. A simple probabilistic calculation shows
\begin{itemize}[noitemsep]
	\item[-] If $q \ll \frac{n}{sk}$, then $\Pr[Q_1 \subseteq B] = 1-o(1)$.
	\item[-] If $q\ll \frac{s^2k}{n}$, then $\Pr[Q_2 \subseteq A] = 1-o(1)$.
\end{itemize}
Another fact we will use is that, conditioned on $Q_1\subseteq B$ and $Q_2\subseteq A$ if the set of vertices queried (as a multiset) has no collisions (whp), then the algorithm can't distinguish between the two cases. 
To see this, one may consider the sampling process as picking a random permutation of the vertices and letting the first $sk$ or $\gamma sk$ vertices be $A$ and the remaining is $B$, and the queries $Q_1$ pick a suffix and $Q_2$ picks a prefix of at most $q$ elements. When $q\leq sk$ (it will be much smaller), these cases lead to the same distribution. More precisely, for any fixed set $S$ of $q$ ids of vertices, the probability they last (or first) $q$ vertices is $S$ is the same in both the YES and the NO case.
\smallskip

\noindent
We now prove the statements, but in the order (c), (a), and then (b).

\begin{asparaitem}
	\item When the algorithm has access to $\FN$, we let 
    $s = n^{1/4}$, $k = \sqrt{n/\rho}$ and notice that $d = \frac{\gamma k s^2}{2n} = \Theta(\rho)$.
    Suppose there is an algorithm making 
    $\ll \min(d, \sqrt[4]{n/d})$ 
    queries.
	Note that the conditions above hold and so we may assume $Q_1 \subseteq B$ and $Q_2 \subseteq A$.
	We now assume that we query $\FN(x)$ for each $x\in Q_1\cup Q_2$. So, we get the matched partner for each $x\in Q_1$, and we get the {\em id} of the clique for every $x\in Q_2$.
	Since $q\ll \sqrt{\gamma k}$, the id's obtained are going to be distinct with $1-o(1)$ probability. Conditioned on this, we see that the queries cannot distinguish between YES and NO better 
	than random chance. In particular, it cannot be a $c$-approximation.
	
	\item Suppose the algorithm has no other oracles. In this case, set 
    $s = \sqrt{n\rho}$, $k = 1$ and notice that $d = \frac{\gamma k s^2}{2n} = \Theta(\rho)$.
    Suppose there is an algorithm making 
    $\ll \min(d, \sqrt{n/d})$ 
    queries.
	Note that the conditions above hold and so we may assume $Q_1 \subseteq B$ and $Q_2 \subseteq A$.
	Next, note that $\RN(x)$ for $x\in Q_1$ gives its partner, and we assume all these are already present in $Q_1$ (it only doubles the size).
	For $y\in Q_2$, $\RN(y)$ returns a random vertex in $y$'s clique. Let $\widehat{Q}_2$ be the subset of vertices obtained after all $\RN(\cdot)$ calls have been made. Note this size is $\ll \min(d, \sqrt{n/d})$.
	Since every clique is $\sqrt{n\rho} = \Theta( d \cdot \sqrt{n/d}) \gg |\widehat{Q}_2|^2$, the set of vertices $\widehat{Q}_2$ has no collisions with $1-o(1)$ probability. Conditioned on this, we see that the queries cannot distinguish between YES and NO better 
	than random chance. In particular, it cannot be a $c$-approximation.

	\item Finally, we consider the case when the algorithm has access to $\PAIR$, we let 
    $s = n^{1/3} \cdot \rho^{2/3}$, $k = \sqrt[3]{n/\rho}$ and notice that $d = \frac{\gamma k s^2}{2n} = \Theta(\rho)$.
    Suppose there is an algorithm making 
    $\ll \min(d, \sqrt[3]{n/d})$ 
    queries.
	As before, we may assume $Q_1 \subseteq B$ and $Q_2\subseteq A$, and as in the previous case, let $\widehat{Q}_2$ be the set of vertices obtained after $\RN(y)$-ing, for $y\in Q_2$.
	As before, since $q\ll \sqrt{s}$, we may assume $\widehat{Q}_2$ is distinct. We now tackle the $\PAIR$ queries.
	Note that $\PAIR(x,z)$ for $x\in B$ is already known for all $z\in Q_1 \cup \widehat{Q_2}$, and so we may assume the algorithm doesn't make these queries.
	Also, we already know $\PAIR(x, \cdot)$ for such an $x$, and so we assume the algorithm doesn't make such queries. Suppose the pair queries are $\{e_1, \ldots, e_t\}$ (where $t\ll k$) where each $e_i = (a_i, b_i)$ for 
	$a_i,b_i \in \widehat{Q}_2$. Note: the algorithm won't query this if $a_i = \RN(b_i)$ or $b_i = \RN(a_i)$ or $a_i, b_i$ are both $\RN(x_i)$ for some other $x_i\in \widehat{Q}_2$ (since it knows the answer).
	We now show that with probability $1-o(1)$, $\PAIR(e_i) = 0$ for {\em all} $e_i$ for both the YES and NO case. This means, that the probability any algorithm can distinguish between the YES and NO case is $o(1)$.
	The proof of the fact is balls-and-bins analysis. The probability $\Pr[\PAIR(e_i) = 1] = \frac{1}{k}$ and so we expect $t/k = o(1)$ pairs to return $1$. The second moment can be shown to be 
	$O(t^2/k^2)$, and thus by Chebyshev we get that $\Pr[\sum_i \PAIR(e_i) \geq 1] = o(1)$.	\qedhere
\end{asparaitem}
\end{proof}

\noindent
We now move to our lower bound arguments for the case when $n$ is unknown. Note that in this case we must also assume that the id of a vertex doesn't reveal any information about $n$.
In particular, we assume the $\id:V\to [U]$ where $U$ is a universe of size larger than any possible size of $|V|$. And in particular, just given a subset of samples with no repeated ids, 
we do not get any information about the size of $V$. 
Let's begin with an easy claim.

\begin{claim}\label{clm:lb:no-re-no-pair}\label{clm:apple}
	Any randomized algorithm that has access to only $\RV$, $\RN$, and $\DEG$ oracles, and 
    outputs a $O(1)$ approximation to the average degree $\davg$ must make $\Omega(\sqrt{n})$ queries.
	This is true irrespective of $\davg$, and in particular, even when $\davg = \Omega(n)$.
\end{claim}
\begin{proof}
	Let $d \in \NN$ be an arbitrary natural number. Suppose the algorithm is $\leq c$-factor approxination for average degree, for some constant $c$.
	We describe two graphs. Graph $G_1$	has $n/d$ cliques with $d$ vertices each, and so $|V_1| = n$ and $|E_1| = nd/2$. Thus, $\davg(G_1) = d$.
	Graph $G_2$ has $n$ cliques with $d$ vertices and one clique with $d\sqrt{cn}$ vertices. So, $|V_2| = nd + d\sqrt{n}$ and $|E_2| = nd^2/2 + cnd^2/2$, and so $\davg(G_2) = c+1$.
	Therefore, the algorithm must be able to distinguish between $G_1$ and $G_2$. 
	
	Now note that the algorithm when run on $G_2$ must return an answer in $\ll \sqrt{|V_1|} = \sqrt{n}$ queries. With $\ll \sqrt{n}$ queries, with all but $o(1)$ probability
	the samples would (a) neither be from the big clique of size $d\sqrt{n}$ since the probability of hitting the big clique is $d\sqrt{cn}/dn = \Omega(\sqrt{n})$, 
	and (b) no vertices will be repeated with either $G_1$ or $G_2$. And so the algorithm cannot distinguish between $G_1$ and $G_2$ and thus can't be a $c$-approximation.
\end{proof}

\noindent
The next claim shows that with $\PAIR$ queries, one can only beat $\sqrt{n}$ if $d\gg \sqrt{n}$.
\begin{claim}\label{clm:lb:pair-no-re}\label{clm:banana}
	Any randomized algorithm which has access to $\RV, \DEG, \RN$ and $\FN$ oracles but {\em not} access to the $\RE$ oracle, and is a $O(1)$ approximation to $\davg$
	must make $\Ot(\sqrt{n/d})$ many queries. 	If we didn't have $\ADD$ but only $\PAIR$ oracle, then the algorithm would need to make $\Omega(\min(\sqrt{n}, n/d))$ queries.
\end{claim}
\begin{proof}
	The instance is the same as in~\Cref{clm:lb:no-re-no-pair}. We know that if $q\ll \sqrt{n}$ then the samples obtained (via $\RV$ and $\RN$) are 
	going to be (a) distinct, and (b) will all have degree $d$. Furthermore, if $q\ll \sqrt{n/d}$, these vertices will all be in different (small) cliques and would be indistinguishable
	for $G_1 $and $G_2$. 
	
	If $q \geq \sqrt{n/d}$, thought, 	
	in the case of $G_1$, some of these samples will be from the same clique, 
	but in $G_2$ they will be all from different cliques. So $\ADD$ queries can distinguish $G_1$ and $G_2$ using $O(\sqrt{n/d})$ queries would distinguish between $G_1$ and $G_2$.
	However, with only $\PAIR$ queries, even when there are vertices from the same clique, if $q\ll n/d$, we cannot distinguish between them. This argument is similar to that in 
	the part (b) of~\Cref{thm:lb-known-n}. Let $\{e_1, \ldots e_t\}$ be the $\PAIR$ queries made. If $t \ll n/d$, then with all but $o(1)$ probability, the answer to all $\PAIR(e_i) = 0$,
	which is what we would have in $G_2$ as well. Therefore, we can't distinguish between $G_1$ and $G_2$ if $q\ll n/d$ as well.  
\end{proof}
\noindent
The next claim shows observation shows that in the unknown $n$ case, no improvement to the above is possible with even when $\RE$ is available.
\begin{claim}\label{clm:cherry}
	For any $d$, any randomized algorithm which has access to all oracles and is a $O(1)$-approximation to $n$ must make $\Omega(\sqrt{n/d})$-oracle calls. 
	If we didn't have $\ADD$ but only $\PAIR$, then the algorithm would need to make $\Omega(\min(\sqrt{n}, n/d)$ queries, and so for $n$-estimation, $\RE$ doesn't help 
	over~\Cref{clm:lb:pair-no-re}. If we didn't even have $\PAIR$ (so basic model plus $\RE$), then we need to make $\Omega(\sqrt{n})$ queries.
\end{claim}
\begin{proof}
	This is simple: consider $G_1$ to be a collection of $n/d$ cliques of size $d$ or $(c+1)n/d$ cliques of size $d$. Note $\RE$ doesn't give any power in such an example.
	If we make $\ll \sqrt{n/d}$ queries, then with all but $o(1)$ probability all vertices are from different cliques, and thus $\FN$ doesn't give any extra power.
	If we had only $\PAIR$ queries instead, then with $\ll n/d$ queries we won't be able to distinguish since with all but $o(1)$ probability
	all $\PAIR$ queries would answer $0$. If we didn't even have $\PAIR$ queries, then unless we query $\sqrt{n}$ vertices we have no distinguish-ability.
\end{proof}

\begin{claim}\label{clm:date}
	Any randomized algorithm which has access to all oracles and is a $O(1)$-approximation to $\davg$ must make 
    $\Omega(\min(d, \sqrt{n/d}))$ queries, which is $\Omega(n^{1/3})$ for worst-case $d$.
\end{claim}
\begin{proof}
	Consider $G_1$ to be a consist of a clique on 
    $\sqrt{n\rho}$
    vertices and a matching of size $n$. 
    Thus, $|V_1| = n + \sqrt{n\rho}$ and $|E_1| = {\sqrt{n\rho} \choose 2} + n$, and thus $\davg(G_1) = \frac{{\sqrt{n\rho} \choose 2} + n}{n} = \Theta(\rho)$. 
    Let $G_2$ be the same but a matching of size $(c+1)n$. Let $Q_1$ and $Q_2$ be the vertices sampled using $\RV$ and $\RE$ oracles. 
	As in the proof of~\Cref{thm:lb-known-n}, when $q\ll \min(d, \sqrt{n/d})$, with all but $o(1)$ probability, $Q_1$ belongs to the matching and $Q_2$ belongs to the clique. Note that conditioned on this $\FN$
	doesn't help. Furthermore, since $q\ll \sqrt{n}$, we don't have any repetitions in $Q_1$ and so we won't be distinguish between $G_1$ and $G_2$. 
\end{proof}

\noindent
We end this section with a simple proof of part (b) in~\Cref{thm:onlyedges} answering a question in~\cite{BT24}
on what can be done with $\RE$, but no $\RV$, queries. 

\begin{claim}\label{clm:elderberry}
	Suppose there is an algorithm which returns a constant factor approximation to $\davg$ only using the $\RE, \RN, \DEG$ and $\PAIR$ oracles. Then it must make $\Omega(\sqrt{n})$ queries
	Instead of $\PAIR$, if it has $\ADD$ queries, then it must make $\Omega(n^{1/3})$ queries. 
\end{claim}
\begin{proof}
	$G_1$ has $\sqrt{n}$ cliques of size $\sqrt{n}$ and a matching of size $cn$. Thus, $|V(G_1)| = (c+1)n$ and $|E(G_1)| = \frac{n\sqrt{n}}{2}(1 + o(1))$ and $\davg(G_1) = \frac{\sqrt{n}}{c+1}(1 + o(1))$.
	$G_2$ has $c\sqrt{n}$ cliques of size $\sqrt{n}$ and a matching of size $n$. Thus, $|V(G_2)| = (c+1)n$ and $|E(G_2)| = \frac{cn\sqrt{n}}{2}(1 + o(1))$ and $\davg(G_2) = \sqrt{n}(1 + o(1))$. 
	If an algorithm is a $c$-approximation, then it should distinguish between the two cases above. 	
	
	With $\RE$ queries with $\ll \sqrt{n}$ queries, if we let $Q_2$ be the endpoints of $\ll \sqrt{n}$ many $\RE$ queries, with $1-o(1)$ probability they will all lie in the cliques.
	Let $\{e_1, \ldots, e_t\}$ be the $\PAIR$ queries made on these sampled nodes. Note that $\Pr[\PAIR(e_i) = 1] = 1/\sqrt{n}$ and so if $t\ll \sqrt{n}$, as multiple times before, we get that with all but $o(1)$ probability
	all $\PAIR(e_i) = 0$ in both $G_1$ and $G_2$. Thus, with $1-o(1)$ probability these are indistinguishable.
	
	The $\ADD$ lower bound constructs $G_1$ (respectively $G_2$) with  $n^{2/3}$ (or $cn^{2/3}$ cliques, respectively) of size $n^{1/3}$, and a matching of size $cn$ (or $n$, respectively).
	Now with $\ll n^{1/3}$ queries all $\RE$ and $\RN$ samples will lie in distinct cliques, and $\ADD$ (or $\FN$) won't give information.
\end{proof}

\bibliographystyle{plain}
\bibliography{sublinear}

\begin{thebibliography}{10}

\bibitem{AdMcMu22}
Raghavendra Addanki, Andrew McGregor, and Cameron Musco.
\newblock Non-adaptive edge counting and sampling via bipartite independent set queries.
\newblock In {\em Proc., European Symposium on Algorithms}, volume 244, pages 2:1--2:16, 2022.

\bibitem{AhNeKo12}
N.~Ahmed, J.~Neville, and R.~Kompella.
\newblock Space-efficient sampling from social activity streams.
\newblock In {\em SIGKDD BigMine}, pages 1--8, 2012.

\bibitem{ANK13}
Nesreen~K Ahmed, Jennifer Neville, and Ramana Kompella.
\newblock Network sampling: From static to streaming graphs.
\newblock {\em ACM Transactions on Knowledge Discovery from Data (TKDD)}, 8(2):1--56, 2013.

\bibitem{ANK14}
Nesreen~K Ahmed, Jennifer Neville, and Ramana Kompella.
\newblock Network sampling: From static to streaming graphs.
\newblock {\em ACM Transactions on Knowledge Discovery from Data (TKDD)}, 8(2):7, 2014.

\bibitem{ANK10}
N.K. Ahmed, J.~Neville, and R.~Kompella.
\newblock Reconsidering the foundations of network sampling.
\newblock In {\em WIN 10}, 2010.

\bibitem{ABGPRY18}
Maryam Aliakbarpour, Amartya~Shankha Biswas, Themis Gouleakis, John Peebles, Ronitt Rubinfeld, and Anak Yodpinyanee.
\newblock Sublinear-time algorithms for counting star subgraphs via edge sampling.
\newblock {\em Algorithmica}, 80:668--697, 2018.

\bibitem{AnCh08}
Dana Angluin and Jiang Chen.
\newblock Learning a hidden graph using o(logn) queries per edge.
\newblock {\em J.\ Comput.\ System Sci.}, 74(4):546--556, 2008.

\bibitem{AsKaKh19}
Sepehr Assadi, Michael Kapralov, and Sanjeev Khanna.
\newblock A simple sublinear-time algorithm for counting arbitrary subgraphs via edge sampling.
\newblock In {\em Proc., Innovations in Theoretical Computer Science (ITCS)}, volume 124, pages 6:1--6:20, 2019.

\bibitem{BG08}
Ziv Bar-Yossef and Maxim Gurevich.
\newblock Random sampling from a search engine's index.
\newblock {\em Journal of the ACM}, 55(5):1--74, 2008.

\bibitem{BKEBS25}
Sabyasachi Basu, Nadia K\={o}shima, Talya Eden, Omri Ben-Eliezer, and C.~Seshadhri.
\newblock A sublinear algorithm for approximate shortest paths in large networks.
\newblock In {\em Proc., ACM Int. Conf. on Web Search and Data Mining (WSDM)}, page 20–29, 2025.

\bibitem{BeHa+20}
Paul Beame, Sariel Har{-}Peled, Sivaramakrishnan~Natarajan Ramamoorthy, Cyrus Rashtchian, and Makrand Sinha.
\newblock Edge estimation with independent set oracles.
\newblock {\em ACM Trans. on Algorithms (TALG)}, 16(4):52:1--52:27, 2020.

\bibitem{BEOF22}
Omri Ben-Eliezer, Talya Eden, Joel Oren, and Dimitris Fotakis.
\newblock Sampling multiple nodes in large networks: Beyond random walks.
\newblock In {\em Proc., ACM Int. Conf. on Web Search and Data Mining (WSDM)}, pages 37--47, 2022.

\bibitem{BT24}
Lorenzo Beretta and Jakub T{\v{e}}tek.
\newblock Better sum estimation via weighted sampling.
\newblock {\em ACM Trans. on Algorithms (TALG)}, 20(3):1--33, 2024.
\newblock \emph{Prelim. version in Proc., SODA 2022}.

\bibitem{BCM25}
Arijit Bishnu, Debarshi Chanda, and Gopinath Mishra.
\newblock Arboricity and random edge queries matter for triangle counting using sublinear queries.
\newblock {\em arXiv preprint arXiv:2502.15379}, 2025.

\bibitem{ChLeWa20}
Xi~Chen, Amit Levi, and Erik Waingarten.
\newblock Nearly optimal edge estimation with independent set queries.
\newblock In {\em Proc., ACM-SIAM Symposium on Discrete Algorithms (SODA)}, pages 2916--2935, 2020.

\bibitem{CDKLS16}
Flavio Chiericetti, Anirban Dasgupta, Ravi Kumar, Silvio Lattanzi, and Tam{\'a}s Sarl{\'o}s.
\newblock On sampling nodes in a network.
\newblock In {\em Proc., International World Wide Web Conference (WWW)}, pages 471--481, 2016.

\bibitem{ChHa18}
Flavio Chierichetti and Shahrzad Haddadan.
\newblock {On the Complexity of Sampling Vertices Uniformly from a Graph}.
\newblock In {\em Proc., International Colloquium on Automata, Languages and Programming (ICALP)}, pages 149:1--149:13, 2018.

\bibitem{DKS14}
Anirban Dasgupta, Ravi Kumar, and Tam{\'a}s Sarl\'os.
\newblock On estimating the average degree.
\newblock In {\em Proc., International World Wide Web Conference (WWW)}, pages 795--806, 2014.

\bibitem{EdJa+18}
Talya Eden, Shweta Jain, Ali Pinar, Dana Ron, and C.~Seshadhri.
\newblock Provable and practical approximations for the degree distribution using sublinear graph samples.
\newblock In {\em Proc., International World Wide Web Conference (WWW)}, pages 449--458, 2018.

\bibitem{EdRoRo19}
Talya Eden, Dana Ron, and Will Rosenbaum.
\newblock {The Arboricity Captures the Complexity of Sampling Edges}.
\newblock In {\em Proc., International Colloquium on Automata, Languages and Programming (ICALP)}, pages 52:1--52:14, 2019.

\bibitem{ERS19}
Talya Eden, Dana Ron, and C~Seshadhri.
\newblock Sublinear time estimation of degree distribution moments: The arboricity connection.
\newblock {\em SIAM Journal on Discrete Mathematics (SIDMA)}, 33(4):2267--2285, 2019.
\newblock \emph{Prelim. version in Proc., ICALP 2017}.

\bibitem{EdRoSe20}
Talya Eden, Dana Ron, and C.~Seshadhri.
\newblock Faster sublinear approximation of the number of \emph{k}-cliques in low-arboricity graphs.
\newblock In {\em Proc., ACM-SIAM Symposium on Discrete Algorithms (SODA)}, pages 1467--1478, 2020.

\bibitem{EdRo18-2}
Talya Eden and Will Rosenbaum.
\newblock Lower bounds for approximating graph parameters via communication complexity.
\newblock In {\em Proc., International Workshop on Randomization and Computation (RANDOM)}, volume 116, pages 11:1--11:18, 2018.

\bibitem{ER18}
Talya Eden and Will Rosenbaum.
\newblock Lower bounds for approximating graph parameters via communication complexity.
\newblock {\em Proc., International Workshop on Randomization and Computation (RANDOM)}, 2018.

\bibitem{EdRo18}
Talya Eden and Will Rosenbaum.
\newblock {On Sampling Edges Almost Uniformly}.
\newblock In {\em Proc. Symposium on Simplicity in Algorithms (SOSA)}, pages 7:1--7:9, 2018.

\bibitem{Feig06}
Uriel Feige.
\newblock On sums of independent random variables with unbounded variance and estimating the average degree in a graph.
\newblock {\em SIAM Journal on Computing (SICOMP)}, 35(4):964--984, 2006.
\newblock \emph{Prelim. version in Proc., STOC 2004}.

\bibitem{G17-book}
O.~Goldreich.
\newblock {\em Introduction to Property Testing}.
\newblock Cambridge University Press, 2017.

\bibitem{GR08}
Oded Goldreich and Dana Ron.
\newblock Approximating average parameters of graphs.
\newblock {\em Random Structures Algorithms}, 32(4):473--493, 2008.
\newblock \emph{Prelim. version in Proc., RANDOM 2006 \& ECCC Tech Report 2004}.

\bibitem{GRS11}
M.~Gonen, D.~Ron, and Y.~Shavitt.
\newblock Counting stars and other small subgraphs in sublinear-time.
\newblock {\em SIAM Journal on Discrete Mathematics (SIDMA)}, 25(3):1365--1411, 2011.

\bibitem{GK00}
Vladimir Grebinski and Gregory Kucherov.
\newblock Optimal reconstruction of graphs under the additive model.
\newblock {\em Algorithmica}, 28(1):104--124, 2000.

\bibitem{HHMN00}
Monika~R Henzinger, Allan Heydon, Michael Mitzenmacher, and Marc Najork.
\newblock On near-uniform url sampling.
\newblock {\em Computer Networks}, 33(1-6):295--308, 2000.

\bibitem{KaLiSo11}
Liran Katzir, Edo Liberty, and Oren Somekh.
\newblock Estimating sizes of social networks via biased sampling.
\newblock In {\em Proc., International World Wide Web Conference (WWW)}, pages 597--606, 2011.

\bibitem{KaKrRo04}
Tali Kaufman, Michael Krivelevich, and Dana Ron.
\newblock Tight bounds for testing bipartiteness in general graphs.
\newblock {\em SIAM Journal on Computing (SICOMP)}, 33(6):1441--1483, 2004.

\bibitem{LF06}
Jure Leskovec and Christos Faloutsos.
\newblock Sampling from large graphs.
\newblock In {\em ACM SIGKDD International Conference on Knowledge Discovery and Data Mining}, pages 631--636, 2006.

\bibitem{LNSS93}
Richard~J Lipton, Jeffrey~F Naughton, Donovan~A Schneider, and Sridhar Seshadri.
\newblock Efficient sampling strategies for relational database operations.
\newblock {\em Theoretical Computer Science}, 116(1):195--226, 1993.

\bibitem{MB11}
A.~S. {Maiya} and T.~Y. {Berger-Wolf}.
\newblock Benefits of bias: Towards better characterization of network sampling.
\newblock In {\em ACM SIGKDD International Conference on Knowledge Discovery and Data Mining}, pages 105--113, 2011.

\bibitem{MU17}
Michael Mitzenmacher and Eli Upfal.
\newblock {\em Probability and computing: Randomization and probabilistic techniques in algorithms and data analysis}.
\newblock Cambridge university press, 2017.

\bibitem{MPX07}
Rajeev Motwani, Rina Panigrahy, and Ying Xu.
\newblock Estimating sum by weighted sampling.
\newblock In {\em Proc., International Colloquium on Automata, Languages and Programming (ICALP)}, pages 53--64, 2007.

\bibitem{PaRo02}
Michal Parnas and Dana Ron.
\newblock Testing the diameter of a graph.
\newblock {\em Random Structures Algorithms}, 20(2):165--183, 2002.

\bibitem{RoTs16}
Dana Ron and Gilad Tsur.
\newblock The power of an example: Hidden set size approximation using group queries and conditional sampling.
\newblock {\em Trans. on Computing Theory}, 8(4):15:1--15:19, 2016.

\bibitem{Sesh15}
C~Seshadhri.
\newblock A simpler sublinear algorithm for approximating the triangle count.
\newblock {\em arXiv preprint arXiv:1505.01927}, 2015.

\bibitem{SEGP16}
Sucheta Soundarajan, Tina Eliassi{-}Rad, Brian Gallagher, and Ali Pinar.
\newblock Maxreach: Reducing network incompleteness through node probes.
\newblock In {\em Int.\ Conf.\ on Adv. in Soc. Networks and Mining (ASONAM)}, pages 152--157, 2016.

\bibitem{SEGP17}
Sucheta Soundarajan, Tina Eliassi{-}Rad, Brian Gallagher, and Ali Pinar.
\newblock epsilon-{WGX:} adaptive edge probing for enhancing incomplete networks.
\newblock In {\em Web Science Conference}, pages 161--170, 2017.

\bibitem{TT22}
Jakub T{\v{e}}tek and Mikkel Thorup.
\newblock Edge sampling and graph parameter estimation via vertex neighborhood accesses.
\newblock In {\em Proc., ACM Symposium on Theory of Computing (STOC)}, pages 1116--1129, 2022.

\bibitem{Wata05}
Osamu Watanabe.
\newblock Sequential sampling techniques for algorithmic learning theory.
\newblock {\em Theoretical Computer Science}, 348(1):3--14, 2005.

\end{thebibliography}

\pagebreak
\appendix
\section{Subroutines from previous works provided for completeness}

\subsection{Beretta-\Tetek average degree estimation using random edge samples}\label{sec:app:bt}

We provide the $O(\eps^{-2}d)$-query algorithm of~\cite{BT24} for completeness.

\begin{algorithm}
	\caption{{\sc \BTEstAvgDeg}($\theta, \eps$)}\label{alg:bt-est-theta}
	\begin{algorithmic}[1]
		\LineComment{Assume access to $\RE$ and $\DEG$, and {\em no knowledge} of $n$. Assumes minimum degree $\geq 1$.}
		\LineComment{Returns: either $\widehat{\davg}$, an estimate to average degree, or asserts $\davg > \theta$}
		\LineComment{Query Complexity: $O(\frac{\theta}{\eps^2})$}
		\Statex
		\State $k \eq \frac{100\theta}{\eps^2}$  \Comment{No attempt has been made to optimize this constant ``100''}
		\For{$i\in \{1,2,3,\ldots, k\}$}:
			\State Sample $e \in E$ and sample $x\in e$ uar. \Comment{Uses $\RE$} \label{albacore}
			\State Set $Z_i \eq \frac{1}{\deg(x)}$
		\EndFor
		\State $Z \eq \frac{1}{k}\sum_{i=1}^k Z_i$ \label{mule}
		\If{$Z < \frac{1}{2\theta}$}:
			\State Assert $\davg > \theta$ \label{bream}
		\Else:
			\State \Return $\widehat{\davg} \eq \frac{1}{Z}$ \label{catfish}
		\EndIf
	\end{algorithmic}
\end{algorithm}

\begin{lemma}
	{\sc \BTEstAvgDeg}($\theta, \eps$) runs in $O(\frac{{\theta}}{\eps^2})$ queries,
	and 
	\begin{itemize}[noitemsep]
		\item If $\theta \geq \davg$, wp $\geq 0.99$ the algorithm returns $Z\in (1\pm \eps)\frac{1}{\davg}$,
		\item For any $i\geq 1$, if $\theta < \frac{\davg}{2^i}$, then $\Pr[Z \geq \frac{1}{2\theta}] \leq \frac{1}{2^{i-1}}$. This is the probability of not asserting.
	\end{itemize}
In particular, wp $\geq 1 - (0.16 + 0.125) > 0.7$, it either returns a good estimate or correctly asserts $\davg>\theta$.
%
	This algorithm doesn't need to know the number of vertices $n$, but assumes the minimum degree is $\geq 1$.	
\end{lemma}
\begin{proof}
	Note that $x$ sampled in~\Cref{albacore} is sampled proportional to its degree. 
	So, it is easy to see $\Exp[Z_i] = \sum_{x\in V} \frac{\deg(x)}{2m} \cdot \frac{1}{\deg(x)} = \frac{1}{\davg}$.
	Using the fact that $\deg(x) \geq 1$, one sees $\Var[Z_i] \leq \frac{1}{\davg}$ as well. And so,
	\[
		\Exp[Z] = \frac{1}{\davg} ~~~\text{and}~~~ \frac{\Var[Z]}{(\Exp[Z])^2} \leq \frac{\eps^2}{100}\cdot \frac{\davg}{\theta}
	\]
	The first bullet point follows from Chebyshev and the second from Markov (\Cref{fact:mkov-cheb}).
%
\end{proof}
\noindent
{\em Remark. If we knew the fraction of non-isolated vertices is $\geq \beta > 0$, then if we replace ``$100$'' by ``$100/\beta$'', the same proof would go through.}\smallskip

\noindent
Using the above, one can obtain the following theorem.
\bertet*

\begin{proof}
	For $\theta \in \{1,2, 4, 8, \ldots\}$, we run {\sc \BTEstAvgDeg}$(\theta, \eps)$ for $\beta = C\ln(1/\delta)$ iterations,
	and if $\geq \beta/2$ of these iterations return a $Z$, we and return the reciprocal of the median value of them.
	Let $\theta^\star$ be the value at which it does. 
	Using the above lemma one can argue that with probability $1-\delta$, we will have $\frac{\davg}{8} \leq \theta^\star \leq \davg$ and 
	furthermore the estimate returned is $\in (1\pm \eps)\frac{1}{\davg}$; the latter follows since $\theta^\star \geq \frac{\davg}{8}$, Chebyshev implies the probability $\Pr[Z \notin (1\pm \eps)\frac{1}{\davg}] < \frac{1}{12}$, which implies
	that the median will be in that interval with $1-\delta$ probability.
\end{proof}

\subsection{Eden-Ron-Seshadhri average degree estimation in the base model}\label{sec:app:ers}

We provide the $O(\eps^{-2}\sqrt{n/d})$-query algorithm of Eden, Ron, Seshadhri~\cite{ERS19} for completeness.
\begin{algorithm}
	\caption{{\sc \ERSEstAvgDeg}($\eps$)}\label{alg:ers}
	\begin{algorithmic}[1]
		\LineComment{Assume access to $\RV$, $\RN$, and $\DEG$, and {\em knowledge} of $n$}
		\LineComment{Returns $\widehat{\davg}$ an estimate to average degree}
		\LineComment{Query Complexity: $O(\eps^{-2}\sqrt{n/d})$}
		\Statex
		\State $\guess \eq n$; $k \eq \frac{100}{\eps^2}\sqrt{\frac{n}{\guess}}$ \Comment{No attempt has been made to optimize this constant ``100''}
		\While{True}:
		\For{$t$ in $\{1,2,\ldots, \ceil{10\ln(1/\delta)}\}$}: \Comment{$\delta$ is failure probability; set $\delta = 0.01$ for concreteness}
			\For{$i$ in $\{1,2,\ldots, k\}$}:
				\State Sample $x \in V$ uar. \Comment{Uses $\RV$}
				\State Sample $y \sim x$ uar. \Comment{Uses $\RN$}
				\If{($\deg(y) \geq \deg(x)$) or (($\deg(y) = \deg(x)$) and $\id(y) > \id(x)$)}
					\State $Z_i \eq 2\deg(x)$
				\Else:
					\State $Z_i \eq 0$
				\EndIf				
			\EndFor
			\State $Z^{(t)} \eq \frac{1}{k}\sum_{i=1}^k Z_i$ \label{line:tiktiki}
		\EndFor
		\State $Z \eq \mathrm{median}(Z^{(1)}, \ldots, Z^{(t)})$
		\If{$Z < \frac{\guess}{2}$}:\label{line:ers:dingo}
			\State $\guess \eq \guess/2$; $k \eq \frac{100}{\eps^2}\sqrt{\frac{n}{\guess}}$ \Comment{No attempt has been made to optimize this constant ``100''}
		\Else:
			\State \Return $\widehat{\davg} \eq Z$
		\EndIf
		\EndWhile
	\end{algorithmic}
\end{algorithm}
\ers*
%
\begin{proof}
	For any $x\in V$, let $\deg^+(x)$ denote the number of neighbors $y$ such that $\deg(y) > \deg(x)$ or $\deg(y) = \deg(x)$ \& $\id(y) > \id(x)$.
	First observe that $\sum_{x\in V} \deg^+(x) = m$. 
	Second, observe that for any $i$ (in any $t$th outer loop), $\Exp[Z_i] = \sum_{x\in V} \frac{1}{n} \frac{\deg^+(x)}{2\deg(x)}\cdot \deg(x) = \davg$.
	Third observe that $\Var[Z_i] \leq \frac{1}{n} \sum_{x\in V} \deg^+(x)\deg(x) \leq \left(\max_{x\in V} \deg^+(x)\right)\davg$, and it is easy to see $\max_{x\in V} \deg^+(x) \leq \sqrt{2m}$. 
	In particular, for all $t$, we have
	\[
		\Exp[Z^{(t)}] = \davg ~~~\text{and}~~~ \frac{\Var[Z^{(t)}]}{(\Exp[Z^{(t)}])^2} \leq \frac{\eps^2}{100}\cdot \sqrt{\frac{\guess}{n}}\cdot \frac{1}{\davg^2} = \frac{\eps^2}{100}\cdot \sqrt{\frac{\guess}{\davg}}
	\]
	\smallskip
	
	\noindent
	Let $\guess^\star$ be the final value of $\guess$. 
	\begin{claim}
		$\Pr[\frac{\davg}{2} \leq \guess^\star \leq 32\davg] > 1-\delta$
	\end{claim}
	\begin{proof}
		Fix an $i\geq 5$ and consider a value of $g \geq 2^i d$. 
 		By Markov's inequality, note that for any $t$, $\Pr[Z^{(t)} \geq 2^{i-1}\davg] \leq \frac{1}{2^{i-1}}$.
		Using Chernoff bounds (\Cref{fact:chern}(c)) and using $t \geq 10\ln(1/\delta)$, one can upper bound the probability of the median being large as 
		$\Pr[Z \geq 2^{i-1}d] \leq \exp(-\frac{10\ln(1/\delta)}{4}\ln(2^{i-2})) \leq \delta^{i}$.
		Therefore, the probability $\guess^* > 32\davg$ is, by union bound, at most $\sum_{i=5}^{\infty} \delta^{i} \ll \delta$.
		
		The probability $\Pr[\guess^\star < d/2]$ is upper bounded
		by the probability that for $\guess < d$, the~\Cref{line:ers:dingo} evaluated to {\em true}. For such a $\guess$, note that for all $t$, $\Pr[Z^{(t)} < \guess/2] \leq \Pr[Z^{(t)} < \frac{\Exp[Z^{(t)}]}{2}]$.
		By Chebyshev this probability is at most $< \frac{4\eps^2}{100} < \frac{1}{25}$. The probability that the median $Z$ is $< \guess/2$ is therefore, by Chernoff bounds (\Cref{fact:chern}(c)), $\leq e^{-\frac{10}{4}\ln(1/\delta)\ln(25/2)} \ll \delta$.
		The claim follows.
	\end{proof}
	\noindent
	Conditioned on $\guess^\star \leq 32\davg$, we get that $\Pr[Z^{(t)} \notin (1\pm \eps)\davg] \leq \frac{1}{100} \sqrt{\frac{32\davg}{\davg}} < 0.06$.
	By Chernoff bounds (\Cref{fact:chern}(c)) the probability the median is not in the range is 
	$\Pr[Z \notin (1\pm \eps)\davg] \leq \exp(-\frac{10\ln(1/\delta)}{4} \ln(\frac{1}{0.12}))$ which is 	$\ll \delta$.
	Conditioned on $\guess^\star \geq \frac{\davg}{2}$, the total query complexity of the algorithm is at most a geometric sum of $O(\eps^{-2}\sqrt{n/\guess})$ where
	$\guess$ ranges as $n$, $n/2$, $n/4$,$\cdots$, and is at least $\frac{\davg}{2}$. This geometric sum is at most $O(\eps^{-2}\sqrt{n/d})$.
\end{proof}

\noindent
\begin{remark}\label{rem:ers}
If instead of knowing $n$, we had an estimate $\frac{n}{2} \leq \widehat{n} \leq 2n$, even then the above analysis would go through with a $O(1)$
hit in the number of queries. In particular, the variance $\Var[Z^{(t)}]$ could be twice as large, and the query complexity would have $n$ replaced by $\widehat{n}$, but that won't change the answer in Big-Oh notation.
On the other hand, if instead we had an estimate $\widehat{\rho}$ of $n/\davg$ with $\frac{n}{2\davg} \leq \widehat{\rho} \leq \frac{2n}{\davg}$, then we can bypass the ``guessing-$\guess$-and-halving'' in the while-loop above 
and just run the for-loop with $k \eq \frac{100}{\eps^2}\ceil{\sqrt{\widehat{\rho}}}$.
\end{remark}

\subsection{Estimating universe size with collisions}\label{app:rt}

This result is Theorem 2.1 of Ron-Tsur~\cite{RoTs16}. We provide this for completeness.

\newcommand{\flag}{\mathsf{flag}}
\begin{algorithm}
	\caption{{\sc EstNumColl}($\eps$)}\label{alg:coll-cnt}
	\begin{algorithmic}[1]
		\LineComment{Assume access to $\RV$ and {\em no knowledge} of $n$}
		\LineComment{Returns $\widehat{n}$ an estimate to $n$}
		\LineComment{Query Complexity: $O(\eps^{-2}\sqrt{n})$}
		\Statex
		\State $Y, Z \eq 0$ \Comment{$Y$ counts how many vertices seen more than once}
		\State Initialize empty dictionary $D$
		\State $\tau \eq \ceil{\frac{12}{\eps^2}\ln(4/\delta)}$ \Comment{$\delta$ is failure probability; set it to $0.01$}
		\While{$Y < \tau$}:
			\State $Z\eq Z+1$
			\State Sample $x\in V$ uar. \Comment{Uses $\RV$.}
			\If{$x\in D$ and $\flag(x) = 0$}:
				\State $Y\eq Y + 1$ \Comment{Increment collision counter}
				\State Set $\flag(x) = 1$
			\Else:
				\State Add $x$ to $D$ and set $\flag(x) = 0$. \Comment{Add $x$ to the dictionary}
			\EndIf
		\EndWhile
		\State \Return $\widehat{n} \eq \frac{Z^2}{2\tau}$
	\end{algorithmic}
\end{algorithm}

\rt*

%
\begin{proof}
	The analysis follows from a balls-and-bins style argument. 
	Fix $t$ to be the largest natural number $t$ such that $\frac{t^2}{2\tau} < (1-\eps)n$ or equivalently $\frac{t^2}{2n} < (1-\eps)\tau$.
	Let $\calB_t$ be the event that sampling $t$ vertices leads to $\tau$ different vertices sampled more than once.
	Then note that $\Pr[\widehat{n} < (1-\eps)n] \leq \Pr[\calB_t]$. 
	To upper bound $\Pr[\calB_t]$, we will use the Poisson approximation method (for instance, see Chapter 5 of Mitzenmacher-Upfal book~\cite{MU17}). Let $Z_1, \ldots, Z_n$ be iid Poisson random variables $\sim \Pois(t/n)$.
	Let $X_i$ be the indicator of the event $\{Z_i \geq 2\}$, and let $\calB'_t$ be the event $X := \sum_{i=1}^n X_i \geq \tau$. This event is monotone in the $Z_i$'s and so we get
		$\Pr[\calB_t] \leq 2\Pr[\calB'_t]$. The latter can be upper bounded using standard Chernoff bounds as follows. 
	First one notes that 
	\[
		\Exp[X_i] = 1 - \exp(-\frac{t}{n})\left(1 + \frac{t}{n}\right) = \frac{t^2}{2n^2}\cdot (1 - o(1))~~\Rightarrow~~\Exp[X] =  \frac{t^2}{2n}\cdot(1 - o(1)) < (1 - \eps)\tau
	\]
	where we use the fact $t/n = o(1)$ by our choice of $t$. Therefore, by Chernoff bounds~\Cref{fact:chern}, we get that 
	\[
		\Pr[\calB_t] \leq 2\Pr[\calB'_t] = 2\Pr\left[X \geq \frac{\Exp[X]}{(1-\eps)}\right] \leq 2e^{-\eps^2(1-\eps)\tau/4} \leq \frac{\delta}{2}
	\]
	since $t > \frac{8\ln(4/\delta)}{\eps^2}$. 
	
	Similarly, let $s$ be the smallest natural number such that $\frac{s^2}{2\tau} > (1+\eps)n$ or equivalently $\frac{s^2}{2n} > (1+\eps)\tau$. Let 
	$\calB_s$ be the event that sampling $s$ vertices {\em doesn't} lead to $\tau$ vertices being seen more than once. We get $\Pr[\widehat{n} > (1+\eps)n] \leq \Pr[\cB_s]$. As above, we define Poisson random variables but now with parameter $\Pois(s/n)$, define $X_i$'s and $X$ similarly, get $\Exp[X] = \frac{s^2}{2n}(1 - o(1)) > (1+\eps/2)\tau$. Now we apply the other tail of Chernoff bound.
	We get
	\[
		\Pr[\calB_s] \leq 2\Pr[\calB'_s] = 2\Pr\left[X < \frac{\Exp[X]}{(1+\eps/2)}\right] \leq 2\Pr\left[X < (1 -\eps/4)\Exp[X] \right] \leq 2e^{-\frac{\eps^2(1+\eps/2)\tau}{12}} < \frac{\delta}{2}
	\]
	since $t \geq \frac{12\ln(4/\delta)}{\eps^2}$. Union bounding the above two says that with probability $\geq (1-\delta)$, 
	we have $\widehat{n} \in (1\pm \eps)n$.	
\end{proof}

\end{document}